\documentclass[11pt]{article}

\usepackage[english]{babel}

\usepackage[letterpaper,top=2cm,bottom=2cm,left=2.5cm,right=2.5cm,marginparwidth=1.75cm]{geometry}

\usepackage{amsmath}
\usepackage{amssymb}
\usepackage{amsthm}
\usepackage{graphicx}
\usepackage[colorlinks=true, allcolors=blue]{hyperref}
\usepackage{xcolor}
\usepackage[normalem]{ulem}
\usepackage{float}

\usepackage[vlined]{algorithm2e}
\DontPrintSemicolon
\SetVlineSkip{1pt}
\NoCaptionOfAlgo
\SetKwIF{If}{ElseIf}{Else}{if}{:}{else if}{else:}{}%
\SetKwFor{For}{for}{:}{}%
\SetKwFor{While}{while}{:}{}%
\SetKwProg{Fn}{}{:}{}
\SetCommentSty{textrm}
\newcommand{\SPACE}{\vskip 0.45ex}%
\newcommand{\SPC}[1]{\vskip #1ex}%

\newtheorem{theorem}{Theorem}[section]

\newtheorem{corollary}[theorem]{Corollary}
\newtheorem{lemma}[theorem]{Lemma}

\theoremstyle{definition}

\newtheorem*{remark}{Remark}

\newcommand{\low}{{\mbox{\it low}}}
\newcommand{\out}{\mbox{\it out}}
\newcommand{\tip}{\mbox{\it tip}}
\newcommand{\pre}{\mbox{\it pre}}
\newcommand{\dfs}{\mbox{\it dfs}}
\newcommand{\dfexp}{\mbox{\it dfe}}
\newcommand{\dfe}{\mbox{\it dfe}}

\newcommand{\first}{\mbox{\it first}}
\newcommand{\previsit}{\mbox{\it previsit}}

\newcommand{\postvisit}{\mbox{\it postvisit}}

\newcommand{\retreat}{\mbox{\it retreat}}
\newcommand{\BACKWARD}{\mbox{\it BACKWARD\/}}

\newcommand{\guard}{\mbox{\it guard}}
\newcommand{\set}{\mbox{\it set}}
\newcommand{\stack}{\mbox{\it stack}}

\newcommand{\Advance}{\mbox{\it advance}}
\newcommand{\FORWARD}{\mbox{\it FORWARD\/}}

\newcommand{\push}{\mbox{\it push}}
\newcommand{\pop}{\mbox{\it pop}}
\newcommand{\Top}{\mbox{\it top}}
\newcommand{\NULL}{\mbox{\it null}}
\newcommand{\true}{\mbox{\it true}}
\newcommand{\false}{\mbox{\it false}}
\newcommand{\next}{\mbox{\it next}}
\newcommand{\visited}{\mbox{\it visited}}
\newcommand{\unvisited}{\mbox{\it unvisited}}
\newcommand{\treeadvance}{\mbox{\it tree-advance}}

\newcommand{\treeretreat}{\mbox{\it tree-retreat}}

\newcommand{\nontreetraverse}{\mbox{\it nontree-traverse}}
\newcommand{\leader}{\mbox{\it leader}}
\newcommand{\scan}{\mbox{\it scan}}

\newcommand{\makeset}{\mbox{\it makeset}}
\newcommand{\find}{\mbox{\it find}}
\newcommand{\unite}{\mbox{\it unite}}

\newcommand{\post}{\mbox{\it post}}

\newcommand{\ptr}{\mbox{\it ptr}}
\newcommand{\lowarc}{\mbox{\it lowarc}}

\newcommand{\start}{\mbox{\it start}}

\newcommand{\Time}{\mbox{\it time}}
\newcommand{\F}{F}
\newcommand{\LL}{L}
\newcommand{\lead}{\mbox{\it lead}}

\newcommand{\arc}{\mbox{\it arc}}

\newcommand{\OO}{\mathrm{O}}

\begin{document}

\title{Finding Strong Components Using Depth-First Search}
\author{Robert E. Tarjan \thanks{Department of Computer Science, Princeton University, NJ, USA.
Intertrust Technologies, Milpitas, CA, USA. E-mail: {\tt ret@princeton.edu}.  Research at Princeton University partially supported by by an innovation research grant from Princeton and a gift from Microsoft.} \and Uri Zwick \thanks{Blavatnik School of Computer Science, Tel Aviv University, Tel Aviv, Israel. E-mail: {\tt zwick@tau.ac.il}. Research supported by ISF grant no.\ 2854/20.}}

\maketitle

\begin{abstract}
We survey three algorithms that use depth-first search to find the strong components of a directed graph in linear time: (1) Tarjan's algorithm; (2) a cycle-finding algorithm; and (3) a bidirectional search algorithm.  
\end{abstract}

\bigskip\hfill
In memory of {\bf Pierre Rosenstiehl}, a great colleague, a good friend, and a big fan of depth-first search.

\section{Introduction}\label{S:intro}

In 1972 the first author~\cite{Tarjan72} presented linear-time algorithms that use depth-first search to solve two fundamental graph problems.  The first, which was developed jointly with John Hopcroft~\cite{HoTa73H}, finds biconnected components in an undirected graph.  The second finds strong components in a directed graph.  Although both algorithms are sequential, the strong components algorithm does two computations concurrently, and its correctness requires this concurrency.  This subtlety makes the algorithm especially intriguing.\footnote{``...Tarjan's algorithm for strong components is, without doubt, the algorithm that I love best.  When I learned of this elegant procedure in 1972, I understood for the first time that data structures can be ``deep.""~\cite[page~6]{Knuth21}. ``The data structures that he devised for this problem fit together in an amazingly beautiful way, so that the quantities you need to look at while exploring a directed graph are always magically at your fingertips.''~\cite{Knuth14}.} 

In this paper we develop \emph{Algorithm~T}, a streamlined version of Tarjan's strong components algorithm, from scratch. We incorporate improvements proposed by others in work following the publication of~\cite{Tarjan72}, and we include a few improvements of our own.  In addition, we study two alternative strong components algorithms that also use depth-first search and run in linear time, \emph{Algorithm~C}, which finds cycles and implicitly or explicitly contracts them, and \emph{Algorithm~B}, which does \emph{two} explorations of the graph, one forward, one backward.  We also investigate non-recursive implementations of depth-first search in general and of the three strong components algorithms in particular.

In a directed graph, two vertices $v$ and~$w$ are \emph{mutually reachable} if there is a path from~$v$ to~$w$ and a path from $w$ to~$v$.  Mutual reachability is an equivalence relation, so it partitions the vertices into equivalence classes, i.e., maximal sets of mutually reachable vertices, called the \emph{strongly connected components} or \emph{strong components}.  Our goal is to find the strong components of a directed graph fast.

For concreteness we shall assume that the input graph $G$ is represented by the set~$V$ of its vertices and, for each vertex~$v$, the set $\out(v)$ of arcs exiting~$v$.  Each arc~$a$ has a field $a.\tip$, the vertex~$a$ enters.  When we discuss low-level implementations, we shall assume that $\out(v)$ is represented as an endogenous\footnote{An \emph{endogenous} data structure is a linked structure in which the nodes of the structure \emph{are} the items stored in it, as opposed to an \emph{exogenous} structure, in which the nodes \emph{contain} the items~\cite{Tarjan83}.} singly-linked list whose first arc is $v.\first$, with each arc~$a$ having a field $a.next$, indicating the next arc after~$a$ on the list $\out(v)$ containing~$a$.  If~$a$ is the last arc on $\out(v)$, $a.next = \NULL$.  With minor changes, the algorithms we present also work if each set $\out(v)$ is represented by an array that stores $a.tip$ for each arc $a \in \out(v)$.  Except in Section~\ref{S-non-recursive-dfs}, we treat the graph representation as read-only.  We denote by $n$ the number of vertices and by $m$ the number of arcs in the graph.  In stating bounds we assume that $n = \OO(m)$, which is true for example if every vertex has at least one entering or exiting arc; vertices without entering or exiting arcs are singleton strong components and are easily identified.  We allow loops (arcs from a vertex to itself) and parallel arcs (multiple arcs from and to the same ordered pair of vertices).

\section{Depth-First Search and Its Properties}\label{S-dfs}

\subsection{Depth-first exploration and depth-first search}\label{S-dfs-definition}

A \emph{depth-first exploration} of a directed graph systematically visits its vertices and traverses its arcs.  Arc traversal is a two-part process, consisting of an \emph{advance} on the arc, followed later, possibly much later, by a \emph{retreat} on the arc.  We call an arc \emph{untraversed} if the advance on it has not yet occurred.  The exploration maintains a \emph{current vertex} $v$, initially $\NULL$.  To do a depth-first exploration, mark all vertices unvisited and all arcs untraversed, and repeat the appropriate one of the following three cases until all vertices are visited and all arcs are traversed:

\begin{itemize}
\item[(i)] The current vertex~$v$ is $\NULL$ and some vertex is unvisited.  Let $v$ be any unvisited vertex.  Visit~$v$.  This initiates a \emph{depth-first search} starting from~$v$, which visits all unvisited vertices reachable from~$v$ and traverses all arcs exiting these vertices.

\item[(ii)] The current vertex~$v$ is non-\NULL\ and has an exiting untraversed arc~$a$.  Choose such an arc, say~$a$ from~$v$ to~$w$, and advance on it.  If $w$ is visited, immediately retreat on~$a$.  If, on the other hand, $w$ is unvisited, make~$a$ the \emph{tree arc} entering $w$, set the current vertex $v$ equal to~$w$, and visit the new $v$.

\item[(iii)] The current vertex~$v$ is non-\NULL\ and has no exiting untraversed arcs.  If $v$ is the start vertex of the current search (it has no entering tree arc), set the current vertex $v$ equal to \NULL.  If, on the other hand, $v$ has an entering tree arc~$a$, say from~$u$, retreat on~$a$ and set the current vertex $v$ equal to~$u$.
\end{itemize}

Depth-first search is a purely local process: It reaches each current vertex from the previous one by advancing or retreating on an arc.  Trémaux (see Lucas~\cite{Lucas82}) proposed depth-first search on an undirected graph as a way to explore a maze, but its origins go back to ancient Greece and the legend of Ariadne, Theseus, and the Minotaur in the maze (\cite{Ovid}, book VIII).  Correctly specifying the algorithm and proving it correct are not entirely straightforward.  K{\"o}nig~\cite{Konig90} discusses Trémaux's method and related ones.   Trakhtenbrot~\cite{Trakhtenbrot63} gives a rigorous definition and proof of correctness of the algorithm.  Both K{\"o}nig and Trakhtenbrot consider only undirected graphs.  The full power of depth-first search as a tool to solve algorithmic graph problems only became clear in the early 1970's.

\begin{lemma}
A depth-first exploration generates a set of trees, one rooted at the start vertex of each depth-first search, whose arcs are the tree arcs defined by the exploration.  Each tree contains the set of vertices visited during the search that starts at the tree root.
\end{lemma}

\begin{proof}
Since a vertex becomes visited only once, it has at most one entering tree arc.  Each new tree arc enters an unvisited vertex, which has no outgoing tree arcs, so there are no cycles of tree arcs.  Each vertex is either the start vertex of a search, in which case it has no entering tree arc and is the root of a tree, or it is not a start vertex, in which case it has an entering tree arc and is not a root.  If~$a$ is a tree arc from~$v$ to~$w$, the search that visits $v$ also visits $w$.    
\end{proof} 

We call the trees generated by a depth-first exploration the \emph{depth-first trees}.  Together they comprise the \emph{depth-first forest} generated by the exploration.  The exploration visits the vertices of each depth-first tree in preorder.  We call these visits the \emph{preorder visits}, or \emph{previsits}.  The exploration implicitly visits each vertex~$v$ a second time, in (iii) when $v$ is the current vertex but it has no outgoing untraversed arcs.  We call these visits the \emph{postorder visits}, or \emph{postvisits}, because they occur in postorder on each depth-first tree.  We assign to each of the $2n$ previsits and postvisits a numeric \emph{time}.  We denote the times of the previsit and postvisit of vertex $v$ by $v.\pre$ and $v.\post$, respectively.  The exact visit times do not matter, only their order: If one visit precedes another, the former must have a smaller time than the latter.  That is, the visit times are in the same order as the corresponding visits during the exploration.

The depth-first forest is in general not unique: It depends on the order in which start vertices are selected for depth-first searches and the order of the advances on the arcs out of each vertex.  But in many if not most applications, and to find strong components in particular, \emph{any} depth-first forest will do.

When discussing the depth-first forest, we use the following terminology: If~$a$ is a tree arc from~$v$ to~$w$, $v$ is the \emph{parent} of~$w$, and~$w$ is a \emph{child} of~$v$.  A vertex~$v$ is an \emph{ancestor} of a vertex~$w$, and~$w$ is a \emph{descendant} of~$v$, if there is a path of tree arcs from~$v$ to~$w$.  Such a path can contain no arcs, so each vertex is both a descendant and an ancestor of itself.  If the path contains at least one arc (that is, $v \neq w$), $v$ is a \emph{proper ancestor} of~$w$, and~$w$ is a \emph{proper descendant} of~$v$.  

A depth-first exploration partitions the non-tree, non-loop arcs into three types.  A non-tree arc~$a$ from~$v$ to~$w$ is a \emph{back arc} if~$w$ is a proper ancestor of~$v$, a \emph{forward arc} if $w$ is a proper descendant of~$v$, a \emph{cross arc} if~$v$ and~$w$ are unrelated; that is, neither is an ancestor of the other.  Cross arcs, and only cross arcs, can exit a vertex in one depth-first tree and enter a vertex in another depth-first tree.

\begin{remark} One could treat loops as either back arcs or forward arcs.  We adopt Knuth's suggestion (private communication, 2022) to treat them as neither, but as their own type.  Strong components are unaffected by the addition or deletion of loops, forward arcs, and parallel arcs, since adding or deleting any of these (except possibly the last member of a set of parallel arcs) does not change reachability.
\end{remark}

A depth-first search explicitly or implicitly maintains a path of tree arcs from the start vertex of the search to the current vertex.  Advancing on a tree arc extends this path by one arc; retreating on a tree arc shortens it by one arc.  We call this path the \emph{current path}. The vertices on the current path are exactly those that have been previsited but not postvisited, plus the current vertex.  When an advance on a tree arc occurs, the new current vertex is unvisited but is immediately previsited.  Just before a retreat on a tree arc occurs, the current vertex is postvisited.  A vertex is on the current path during the interval from just before it is previsited until just after it is postvisited, and at no other time. 

The arcs on the current path are exactly the tree arcs on which advances but not retreats have occurred.  There is at most one non-tree arc on which an advance but not a retreat has occurred.  If there is such an arc, say~$a$, $a$ exits the current vertex.  The retreat on~$a$ occurs immediately after the advance on~$a$.

\begin{lemma}\label{L:interval}
Each vertex $v$ is on the current path during the interval between its previsit and its postvisit, each of which happens once.  During this interval, the search advances and retreats on each arc~$a$ exiting~$v$.
\end{lemma}

\begin{proof}
A vertex~$v$ is previsited when it becomes the current vertex and joins the current path. It remains on the current path until it is postvisited. A vertex~$v$ cannot be postvisited until an advance has occurred on every arc~$a$ exiting~$v$. During each such advance, $v$ is the current vertex.  Suppose an advance occurs on such an arc, say~$a$ from~$v$ to~$w$. If $w$ is visited when the advance occurs, a retreat on it occurs immediately.  If not, $a$ becomes a tree arc and $w$ becomes the current vertex.  Vertex $v$ cannot be postvisited until it again becomes the current vertex, which cannot happen until~$w$ is postvisited.  A retreat on~$a$ occurs immediately after~$w$ is postvisited and hence before~$v$ is postvisited.
\end{proof}

\begin{lemma}\label{L:descendants} The following four conditions are equivalent:
\begin{itemize}
\item[\textup{(i)}]Vertex $w$ is a descendant of $v$
\item[\textup{(ii)}]$v.\pre \leq w.\pre < v.\post$
\item[\textup{(iii)}] $v.\pre < w.\post \leq v.\post$
\item[\textup{(iv)}] $v.\pre \leq w.\pre < w.\post \leq v.\post$
\end{itemize}
\end{lemma}

\begin{proof}
If $w$ is a descendant of $v$, then $w$ must be added to and deleted from the current path while~$v$ is on the current path.  It follows from Lemma~\ref{L:interval} that $v.\pre \leq w.\pre < w.\post \leq v.\post$.  Hence (i) implies (ii), (iii), and (iv).

If $v.\pre \leq w.\pre < v.\post$, then $w$ is previsited while~$v$ is on the current path, so $w$ is a descendant of~$v$; that is, (ii) implies (i).  Similarly, if $v.\pre < w.\post \leq v.\post$, $w$ is postvisited while~$v$ is on the current path, so $w$ is a descendant of~$v$; that is, (iii) implies (i).  Finally, (iv) implies (ii) (and (iii)) and hence (i).
\end{proof}


\begin{lemma}\label{L:white-path}
A vertex $w$ is a descendant of a vertex $v$ if and only if there is a path of unvisited vertices from~$v$ to~$w$ just before $v$ is previsited.
\end{lemma}

\begin{proof}
If $w$ is a descendant of $v$, then by Lemma~\ref{L:descendants} all vertices on the path of tree arcs from~$v$ to~$w$ are unvisited just before~$v$ is previsited.  Suppose there is a path $P$ of unvisited vertices from~$v$ to~$w$ just before~$v$ is previsited.  If~$a$ from~$x$ to~$y$ is an arc on $P$ and~$x$ is a descendant of~$v$, then~$y$ is previsited before~$x$ is postvisited by Lemma~\ref{L:interval} and hence before~$v$ is postvisited by Lemma~\ref{L:descendants}.  Since~$y$ is previsited after $v$, $y$ is a descendant of~$v$.  It follows by induction on the number of vertices on $P$ that all of its vertices are descendants of~$v$, since $v$ is a descendant of itself.
\end{proof}

\begin{remark}
Cormen et al.~\cite{CLRS09} call the preceding lemma the \emph{white path theorem}.
\end{remark}


\begin{lemma}\label{L:descendant-exits-entries}
Let $v$ be a vertex and suppose there is an arc~$a$ from~$x$ to~$y$.  \textup{(i)} If~$x$ but not~$y$ is a descendant of~$v$, then $y.\pre < v.\pre$.  \textup{(ii)} If~$y$ but not~$x$ is a descendant of~$v$, then $v.\post < x.\post$.
\end{lemma}

\begin{proof}
This lemma is a corollary of Lemma~\ref{L:descendants}.  (i) Suppose~$x$ is a descendant of~$v$ and $v.\pre \leq y.\pre$.  Since~$x$ is a descendant of~$v$, the advance on~$a$ and hence the previsit of~$y$ occurs before the postvisit of~$v$. That is, $y.\pre < v.\post$.  This inequality and $v.\pre \leq y.\pre$ imply that~$y$ is a descendant of~$v$ by Lemma~\ref{L:descendants}. (ii) Suppose~$y$ is a descendant of~$v$ and $x.\post \leq v.\post$.  Since~$a$ is traversed before~$x$ is postvisited, $y.\pre < x.\post \leq v.\post$.  Since~$y$ is a descendant of~$v$, $v.\pre \leq y.\pre$.  These inequalities combine to give $v.\pre < x.\post \leq v.\post$.  By Lemma~\ref{L:descendants}, $x$ is a descendant of~$v$.      
\end{proof}

\begin{lemma}\label{L:cross-arcs}
If there is a cross arc from~$x$ to~$y$, then $y.\post < x.\pre$. 
\end{lemma}

\begin{proof}
Since~$y$ is not a descendant of~$x$, Lemma~\ref{L:descendant-exits-entries}(i) with $v = x$ implies $y.\pre < x.\pre$.  Given that $y.\pre < x.\pre$, $x.\pre < y.\post$ would imply that~$x$ is a descendant of~$y$ by Lemma~\ref{L:descendants}, a contradiction.  Thus $y.\post < x.\pre$.     
\end{proof}

\subsection{Recursive view of depth-first search}\label{sub-recursive-dfs}

A natural way to express depth-first search is recursively: A depth-first search from an unvisited vertex~$v$ first previsits $v$, second \emph{scans} $v$, and third postvisits $v$.  To scan $v$, it processes each arc~$a$ exiting $v$.  First, it advances on~$a$; second, if $w = a.\tip$ is unvisited, it does a recursive depth-first search from $w$; third, it retreats on~$a$.  That is, an advance on a tree arc from~$v$ to~$w$ suspends the scan of~$v$ and starts a recursive depth-first search from $w$.  Only when the recursive search from $w$ is complete does the scan of~$v$ resume.  

Figure~\ref{F-dfs} implements depth-first exploration using recursive depth-first search. The main function, $\dfe(V)$, does a depth-first exploration of the graph whose vertex set is~$V$ and such that $\out(v)$ is the set of outgoing arcs of~$v$ for each $v \in V$.  The pseudocode contains stubs for the critical events during the exploration: $\start(V)$,  which does any initialization needed for the exploration; $\unvisited(v)$, a Boolean function that returns $\true$ if and only if vertex $v$ is unvisited; $\previsit(v)$, the previsit of~$v$; $\postvisit(v)$, the postvisit of~$v$; $\Advance(v, a, w)$, the advance on an arc~$a$ from~$v$ to~$w$;  $retreat(v, a, w)$, the retreat on an arc~$a$ from~$v$ to~$w$; $\treeadvance(v, a, w)$ and $\treeretreat(v, a, w)$, versions of advance and retreat restricted to tree arcs, and $\nontreetraverse(v, a, w)$, the advance and retreat on a non-tree arc~$a$ from~$v$ to~$w$.  There is no need for separate advance and retreat functions on non-tree arcs since the retreat on a non-tree arc immediately follows the advance on the arc.

Parameter $w$ in the arc traversal stubs is redundant, since $w = a.\tip$, but parameter $v$ is not.  To provide generality, we include both ends of an arc, as well as the arc itself, as parameters.  In an actual application, one would instantiate the stubs with appropriate computations, by defining each stub as a function or as a macro.  Making the stubs macros, so that they are replaced in-line by the appropriate computations, is more efficient, but making them functions allows the generic depth-first search implementation to be shared among several applications. In either case the application can ignore unused parameters in the stubs. The only stubs that must be non-empty are $\start(V)$, which must mark every vertex unvisited; $\previsit(v)$, which must mark~$v$ visited; and $\unvisited(v)$ which must correctly test whether $v$ is unvisited.  The application can do the same or different computations during the advances and retreats on tree and non-tree arcs, by instantiating the appropriate stubs.  Any computation that is the same for advances on tree and non-tree arcs can be done in the $\Advance$ stub, and similarly for retreats.  Algorithm~T, Tarjan's strong components algorithm, instantiates $\start$, $\unvisited$, $\previsit$, $\postvisit$, and $\retreat$, and omits the other stubs.  It does not use parameter~$a$ in $\retreat$, only $v$ and $w$.

The pseudocode includes model implementations of $\start$, $\visited$, $\previsit$, and $\postvisit$ that assign consecutive integer times to all $2n$ vertex visits starting from $1$.  $\start(V)$ initializes the previsit times to~$0$ and sets $\Time$, a global variable, to zero.  Each previsit or postvisit of a vertex $v$ increments $\Time$ and assigns the new value of $\Time$ to $v.\pre$ or $v.\post$, respectively.  A vertex $v$ is unvisited if and only if $v.\pre = 0$.  A simpler way to implement $\unvisited$ is to maintain a bit for each vertex that is initialized to $\false$ and set to $\true$ when the vertex is previsited.  An intermediate solution, which is what two of the strong components algorithms use, is to assign consecutive integer times to just the previsits.

The implementation in Figure~\ref{F-dfs} does not does not explicitly construct the depth-first forest, but it is easy to augment it to do so.  For example, each call of $\treeadvance(v, a, w)$ can set the parent of~$w$ equal to~$v$, or add~$a$ to a list of tree arcs out of~$v$, or both, depending on the application. 

\begin{figure}
\centering
\parbox{2in}{
\begin{algorithm}[H]
\Fn{$\dfexp(V)$}
{
    $\start(V)$ \;
    \For{$s\in V$}
    {\If{$\unvisited(s)$}{$\dfs(s)$ \;}}
}
\end{algorithm}
\vspace*{10pt}
\SetNoFillComment
\begin{algorithm}[H]
\Fn{$\dfs(v)$}{
    $\previsit(v)$ \;
    $\scan(v)$ \;
    $\postvisit(v)$ \;
} 
\end{algorithm}
}
\hspace*{-0.3in}
\parbox{2.8in}{
\begin{algorithm}[H]
\Fn{$\scan(v)$}
{
\For{$a\in \out(v)$}
{
    $w\gets a.\tip$ \;
    $\Advance(v,a,w)$ \;
    \SPACE 
    \eIf{$\unvisited(w)$}
    {
        $\treeadvance(v, a, w)$ \;
        $\dfs(w)$ \;
        $\treeretreat(v, a, w)$ \;
    }
    {   
        $\nontreetraverse(v, a, w)$ \;
    }
    \SPC{0.6} 
    $\retreat(v, a, w)$ \;
}
}
\end{algorithm}
}
\hspace*{-0.3in}
\parbox{1.95in}{
\begin{algorithm}[H]
\Fn{$\start(V)$}
{
    $\Time\gets 0$ \;
    \For{$v\in V$}
    {$v.\pre \gets 0$ \;}
}
\end{algorithm}
\vspace*{2pt}
\begin{algorithm}[H]
\Fn{$\unvisited(v)$}
{
    \Return $v.\pre=0$ \;
}
\end{algorithm}
\vspace*{2pt}
\begin{algorithm}[H]
\Fn{$\previsit(v)$}
{
    $\Time\gets \Time+1$ \;
    $v.\pre\gets \Time$ \;
}
\end{algorithm}
\vspace*{2pt}
\begin{algorithm}[H]
\Fn{$\postvisit(v)$}
{
    $\Time\gets \Time+1$ \;
    $v.\post\gets \Time$ \;
}
\end{algorithm}
}
\caption{A recursive implementation of a generic depth-first exploration.}\label{F-dfs}
\end{figure}

\section{Depth-first search and strong components}\label{S:dfs-and-sccs}

Not only does a depth-first exploration build a spanning forest with certain properties, it also imposes a structure on the strong components.  We exploit this structure to find these components fast.  Note that forward arcs do not affect reachability, so their deletion or addition does not change the strong components. 

\subsection{Properties of strong components}\label{sub-properties}

\begin{lemma}\label{L:strong-component-paths}
If $v$ and $w$ are in the same strong component~$S$, any path from~$v$ to~$w$ contains only vertices in~$S$.
\end{lemma}

\begin{proof}
If $v$ and $w$ are in~$S$, there is a path $P_1$ from $w$ to~$v$.  If~$x$ is on a path $P_2$ from~$v$ to~$w$, $x$ is on the cycle formed by $P_1$ followed by $P_2$ and hence is in~$S$.
\end{proof}

\begin{lemma}\label{L:strong-component-leaders}
In a strong component~$S$, the vertex $v$ with $v.\pre$ minimum is also the vertex with $v.\post$ maximum, and this vertex is an ancestor of every vertex in~$S$.
\end{lemma}

\begin{proof}
This lemma is a corollary of Lemma~\ref{L:white-path}.  The vertex $v$ in~$S$ with $v.\pre$ minimum is the first vertex in~$S$ to be visited. By Lemma~\ref{L:white-path}, $v$ is an ancestor of all vertices in~$S$. By Lemma~\ref{L:descendants}, $v$ is also the vertex in~$S$ with $v.\post$ maximum.
\end{proof}

We define the \emph{leader} of a strong component to be its vertex $v$ with $v.\pre$ minimum (and $v.\post$ maximum).  The other vertices in the component are the \emph{followers} of the leader.  The leaders of the strong components characterize these components, as we shall now show.  Every root of a depth-first tree is the leader of a strong component, but component leaders need not be tree roots.

\begin{lemma}\label{L:follower-parents}
If~$y$ is not the leader of its strong component~$S$, then the parent~$x$ of~$y$ is also in~$S$. 
\end{lemma}

\begin{proof}
Let $v \neq y$ be the leader of~$S$.  There is a path from~$x$ to~$v$ consisting of the tree arc from~$x$ to~$y$ followed by a path from~$y$ to~$v$.  By Lemma~\ref{L:strong-component-leaders} there is a path of tree arcs from~$v$ to~$y$.  This path contains~$x$.  Hence~$x$ is in~$S$. 
\end{proof}

\begin{lemma}\label{L:strong-component-trees}
The strong components partition the depth-first trees into (partial) subtrees: If we delete from the depth-first forest every tree arc entering a strong component leader, the result is a collection of trees, each rooted at a component leader and whose vertices are the leader and its followers.
\end{lemma}

\begin{proof}
The lemma follows immediately from Lemma~\ref{L:follower-parents}.
\end{proof}

\begin{theorem}\label{T:component-top-order}
Ordering the strong components in decreasing postorder of their leaders gives a topological order of the components: If~$a$ is an arc from~$x$ to~$y$, and $u$ and $v$ are the leaders of the strong components containing~$x$ and~$y$, respectively, $u.\post \geq v.\post$.    
\end{theorem}

\begin{proof}
If~$x$ is a descendant of~$v$, then~$x$ and~$y$ are in the same strong component, because there is a cycle consisting of~$a$ followed by a path from~$y$ to~$v$ followed by the path of tree arcs from~$v$ to~$x$.  Hence $u=v$ and the lemma is true.  If~$x$ is not a descendant of~$v$, $u.\post \geq x.\post > v.\post$ by Lemmas~\ref{L:descendants} and~\ref{L:descendant-exits-entries}.
\end{proof}

Our next lemma characterizes strong component leaders. 

\begin{lemma}\label{L:follower-path}
A vertex~$v$ is the leader of a strong component if and only if within its strong component there is no path from~$v$ consisting of zero or more tree arcs followed by one non-tree arc to a vertex~$w$ with $w.\pre<v.\pre$. 
\end{lemma}

\begin{proof} If $v$ is the leader of a strong component, there is \emph{no} path from~$v$ within the component to a vertex $w$ with $w.\pre<v.\pre$.  If $v$ is not the leader of its strong component, then there is a path within the component from~$v$ to the component leader, say $u$, and $u.\pre<v.\pre$.  Let~$y$ be the first vertex on this path that is not a descendant of~$v$, and let~$a$ from~$x$ to~$y$ be the arc on this path that enters~$y$.  Since~$x$ is a descendant of~$v$, there is a path of zero or more tree arcs from~$v$ to~$x$.  Since~$y$ is not a descendant of~$v$, $y.\pre < v.\pre$ and~$a$ is a non-tree arc by Lemma~\ref{L:descendant-exits-entries}. \end{proof}

In the statement of Lemma~\ref{L:follower-path} the requirement ``within the component" is critical; without it the lemma is false.  The crux of Algorithm~T is the ability to test this requirement efficiently. 

The exact structure of the path in the proof of the lemma is less important than the behavior of the depth-first exploration on the path.  We call a path \emph{retreating} if the exploration retreats on its arcs in the reverse of the order of the arcs along the path.  As a special case, a path of no arcs from a vertex to itself is a retreating path.  Any path $P$ of zero or more tree arcs followed by one non-tree arc~$a$ from~$x$ to~$y$ is retreating, because when the advance on~$a$ occurs $P$ excluding~$a$ is part or all of the current path, and~$x$ is the current vertex.  Among the arcs on $P$, the retreat on~$a$ occurs first, followed later by retreats on the tree arcs on $P$ in the reverse of their order on $P$.  Thus the path constructed in the proof of Lemma~\ref{L:follower-path} is retreating, giving us the following corollary:

\begin{lemma}\label{L:low-path}
A vertex~$v$ is a leader of a strong component if and only if within its component there is no retreating path from~$v$ to a vertex $w$ with $w.\pre<v.\pre$. 
\end{lemma}

\subsection{Finding leaders and components}

We shall use Lemma~\ref{L:low-path} to identify strong component leaders.  We define the previsit times $v.\pre$ to be  consecutive integers starting from $1$.  (The strong components algorithm does not use postvisit times, but these can be chosen to be rational numbers so that all previsits and postvisits are correctly ordered by time.\footnote{Define the time of a postvisit to be the (integer) time of the preceding previsit plus $1/(n+1)$ multiplied by the number of postvisits so far.})   We compute, for each vertex~$v$, a value $v.\low$ equal to the minimum $w.\pre$ such that $w$ is reachable from~$v$ by a retreating path within the strong component containing $v$.  To compute the $\low$ values, we initialize $v.\low \gets v.\pre$ when $v$ is previsited.  We update $v.\low$ when retreating on an arc~$a$ from~$v$ to~$w$, by setting $v.\low \gets \min\{v.\low, w.low\}$ \emph{if} $v$ and~$w$ are in the same component.  When $v$ is postvisited, retreats have already occurred on all arcs on retreating paths that start from~$v$, so the computation of $v.\low$ is complete.  Thus $v$ is the leader of a strong component if, and only if, $v.\low = v.\pre$ when $v$ is postvisited.

Of course we cannot \emph{yet} actually do this computation, because we have no way yet of testing whether two vertices are in the same component, so we have no way of knowing whether to update $v.\low$ when retreating on an arc exiting $v$.  But if we find that some vertex $v$ is the leader of a strong component, and we can find the set of vertices in its component, then we can exclude these vertices from computations of $\low$ values of vertices in other strong components.  We do this by setting the $\low$ values of all vertices in the just-found component to~$\infty$ (or to a sufficiently large number).  Then these values do not contribute to later $\low$ updates.  This allows us to update $v.\low$ when retreating on an arc exiting $v$ unconditionally, without explicitly testing whether the ends of the arc are in the same strong component.

Let us restate this computation and prove that it is correct.  When previsiting a vertex~$v$, we set $v.\pre$ equal to the next available natural number, and we initialize $v.\low \gets v.\pre$.  When retreating on an arc from~$v$ to~$w$, we set $v.\low \gets \min\{v.\low, w.\low\}$.  When postvisiting a vertex~$v$, if $v.\low = v.\pre$ we set $w.\low \gets \infty$ for each vertex~$w$ in the strong component whose leader is $v$.  We shall explain later how to find the vertices in the strong component whose leader is $v$.  Excluding this important detail, this is the \emph{$\low$ computation}.

\begin{lemma}\label{L:low-values-correct}
The $\low$ computation maintains the following invariant: After $v$ is previsited but before the leader $u$ of the strong component containing $v$ is postvisited, $v.\low$ is the minimum $w.\pre$ of a vertex $w$ reachable from~$v$ by a retreating path in the same component as $v$ on all of whose arcs a retreat has occurred; after $u$ is postvisited, $v.\low = \infty$.
\end{lemma}

\begin{proof}
We prove the lemma by induction on the number of steps of the exploration.  Only previsits, postvisits, and retreats affect the invariant.  Suppose the invariant holds until~$v$ is previsited.  The previsit correctly initializes $v.\pre$ and correctly sets $v.\low$ to $v.\pre$, preserving the invariant.  Suppose the invariant holds until~$v$ is postvisited.  By Lemma~\ref{L:low-path}, $v$ is the leader of its strong component if and only if $v.\low = v.\pre$ when $v$ is postvisited.  If $v$ is not the leader of a strong component, the postvisit changes no $\low$ values, preserving the invariant.  If $v$ is a leader, the postvisit sets $w.\low \gets \infty$ for every vertex~$w$ in the strong component with leader $v$, also preserving the invariant.

Suppose the invariant holds until the retreat on an arc~$a$ from~$x$ to~$y$.  Suppose~$x$ and~$y$ are in the same strong component, say~$S$.  By Lemma~\ref{L:strong-component-leaders}, the leader of~$S$ is postvisited after~$x$ is postvisited and hence after the retreat on~$a$ occurs, so just before this retreat occurs the values of $x.\low$ and $y.\low$ are respectively the minimum $w.\pre$ such that there is a retreating path in~$S$ from~$x$ or~$y$ on all of whose arcs retreats have already occurred.  The retreat on~$a$ adds to the set of paths that affect the invariant every path consisting of~$a$ followed by a retreating path in~$S$ from~$y$ on all of whose arcs retreats have already occurred.  It follows that the update $x.\low \gets \min\{x.\low, y.\low\}$ correctly preserves the invariant.

Suppose on the other hand that~$x$ and~$y$ are not in the same strong component.  Then retreating on~$a$ produces no new  retreating paths that affect the invariant.  We shall show that $y.\low =\infty$ when the retreat on~$a$ occurs.  This implies that the update $x.\low \gets \min\{x.\low, y.\low\}$ does not change $x.\low$ and hence preserves the invariant.  We do a case analysis based on the type of arc~$a$.  Since~$x$ and~$y$ are in different components, $a$ is not a back arc, because if it were the path of tree arcs from~$y$ to~$x$ followed by arc~$a$ would form a cycle containing~$x$ and~$y$.

Let $u$ be the leader of the strong component containing~$y$.  If~$a$ is a tree arc, $u = y$ by Lemma~\ref{L:strong-component-trees}.  Hence $u$ is postvisited and $y.\low$ is set to~$\infty$ before the retreat on~$a$.

If~$a$ is a forward arc, $u$ is an ancestor of~$y$ and a proper descendant of~$x$ by Lemma~\ref{L:strong-component-trees}.  Since~$y$ is previsited before the advance on~$a$, this advance cannot occur until~$x$ becomes the current vertex after~$y$ is previsited, which cannot occur until all vertices except~$x$ on the path of tree arcs from~$x$ to~$y$, including $u$, have been postvisited.  Thus $y.\low$ is set to~$\infty$ before the advance on~$a$, and hence before the retreat on~$a$.

If~$a$ is a cross arc, $u.\pre \leq y.\pre < x.\pre$ by Lemmas~\ref{L:strong-component-leaders} and~\ref{L:cross-arcs}.  It cannot be the case that $u.\post \geq x.\pre$, for then~$x$ would be a descendant of~$u$ by Lemma~\ref{L:descendants}, and~$x$ and~$y$ would be in the same strong component because of the cycle consisting of the path of tree arcs from~$u$ to~$x$ followed by~$a$ followed by a path from~$y$ to~$u$.  Thus $u.\post < x.\pre$.  In this case also $y.\low$ is set to~$\infty$ before the advance on~$a$, and hence before the retreat on~$a$. 
\end{proof}

What remains is to provide a way to find the vertices in the strong component with leader $v$ when $v$ is postvisited.  For this we use a stack~$F$ of vertices that when postvisited are found to be followers.  When we postvisit a vertex~$v$, we test whether $v$ is a leader (whether $v.\low = v.\pre)$.  If $v$ is not a leader, we push $v$ onto~$F$.  If $v$ is a leader, we test the $\low$ value of the top vertex on~$F$.  If it is at least $v.\low$, we pop this vertex from~$F$ and set its $\low$ value to~$\infty$.  We repeat this until~$F$ is empty or the top vertex on~$F$ has $\low$ value less than $v.\low$.  Finally, we set $v.\low \gets \infty$.  The vertices popped from~$F$, plus $v$, are exactly the vertices in the strong component with leader $v$, as we shall show.  This is the \emph{component-finding computation}.  We can build a representation of each component as part of this computation.

\begin{lemma}\label{L:component-finding-correct}
Suppose that for each vertex~$v$, $v.\low$ has the correct value when $v$ is postvisited.  Then the component-finding computation maintains the invariant that $v.\low = \infty$ if and only if the leader of the component containing $v$ has been postvisited. 
\end{lemma}

\begin{proof}
Suppose the hypothesis of the lemma holds.  The proof of the invariant is by induction on the number of postvisits of leaders.  Suppose the invariant holds until the postvisit of leader $v$.  When $v$ is previsited, every previously postvisited vertex $w$ still on~$F$ has $w.\low < v.\pre$, since $w$ was previsited before $v$ and hence has $w.\low \leq w.\pre < v.\pre$.  Every vertex added to the stack between the previsit and the postvisit of~$v$ is a descendant of~$v$.  Since $v$ is the leader of a strong component, each such vertex $w$ has $w.\low \geq v.\pre$ when $w$ is added to the stack.  These vertices are the followers of~$v$ and of leaders that are proper descendants of~$v$.  By the induction hypothesis, the latter are popped from the stack when their leaders are postvisited.  Those that remain on the stack when $v$ is postvisited are the followers of~$v$ by Lemma~\ref{L:strong-component-trees}, and they are on top of any other vertices on the stack.  The component-finding computation pops exactly these vertices and sets their $\low$ values, and that of~$v$, to~$\infty$, preserving the invariant.      
\end{proof}

Algorithm~T, Tarjan's strong components algorithm, does the $\low$ computation and the component-finding computation during a \emph{single} depth-first exploration.  Neither of the two computations works on its own, nor can they be done one after the other.  This is the subtlety, and the beauty, of the algorithm.\footnote{Even though proving the algorithm correct requires showing that it maintains two invariants at the same time, the algorithm has been formally proved correct using three different formal proof-checking systems~\cite{CCLMT19}.}

\begin{theorem}
Algorithm~T correctly finds the strong components and runs in $\OO(m)$ time.
\end{theorem}

\begin{proof}
The algorithm maintains the invariants in Lemma~\ref{L:low-values-correct} and Lemma~\ref{L:component-finding-correct}.  Hence it correctly finds the strong components.  The algorithm spends $\OO(1)$ time per vertex and arc and thus runs in $\OO(m)$ time.
\end{proof}

By Theorem~\ref{T:component-top-order}, Algorithm~T finds the strong components in a reverse topological order.  Also, the algorithm adds vertices to each component in postorder. In applications that need the components in topological order or need the vertices within each component in postorder, Algorithm~T can provide this information.

\subsection{Extensions}

Algorithms for finding strong components, including the three algorithms we present in this paper, can be extended to compute additional information beyond just the partition of the vertices into strong components.  Two such extensions are to find the condensation and to find spanning in-trees and out-trees of the strong components.  Here we describe how to extend Algorithm~T to do these additional computations.  Other algorithms can be extended similarly.

The \emph{condensation} of a directed graph is the graph formed by contracting each strong component to a single vertex and deleting loops and all but one copy of each set of parallel arcs.  This concept appears in Aho et al.~\cite{AGU72}, although they called the condensation the ``equivalent acyclic graph."  The condensation is acyclic, and its vertices can be ordered topologically, so that each arc leads from a smaller to a larger vertex.

We can extend Algorithm~T to form the condensation, at the cost of an extra pass through the arcs.  We use an extra bit per vertex, initially false.  We represent each vertex of the condensation by the leader of the corresponding strong component.  Since Algorithm~T generates the components in reverse topological order, each arc from a vertex in the component being constructed has its tip in this component or in a previously constructed component.  When adding a vertex $v$ to a component with leader $x$, we examine each arc $a$ out of $v$.  If $v.tip$ is in a previously constructed component, we find $y$, the leader of the component containing $v.tip$.  If the bit of $y$ is false, we set it to true and add to the condensation an arc from $x$ to $y$.  Once the component with leader $x$ is constructed, we do a pass through the arcs out of $x$ in the condensation to reset to false the bits of the tips of these arcs.

Knuth~\cite{Knuth93} gives an implementation of a version of Algorithm~T that produces a representation of the strong components, the condensation, and enough information to verify them.

By Lemma~\ref{L:strong-component-trees}, one can obtain a spanning (out-)tree for each strong component by deleting from the depth-first forest each tree arc entering the leader of a component.  (The current implementation does not maintain tree arcs but it can be easily modified to do so.) We can extend the $\low$ computation in Algorithm~T to compute a spanning in-tree\footnote{An \emph{in-tree} is a set of arcs, one out of each vertex except one, the \emph{root} of the in-tree, such that the arcs form no cycles.} of each strong component.  In addition to computing $v.\low$ for each vertex $v$, we maintain for each follower $v$ a \emph{low arc} $v.\lowarc$, which is the arc whose retreat caused the most recent decrease of $v.\low$. Once the exploration is complete, $v.\lowarc$ is the arc whose retreat gave $v.\low$ its final (minimum) value; it is the first arc on the retreating path whose discovery gave  $v.\low$ its final value.

\begin{lemma}
There is no cycle of low arcs.
\end{lemma}

\begin{proof}
Suppose the lemma is false; that is, there is a cycle of low arcs.  Every vertex on the cycle is a follower, since no leader has a low arc, and all vertices on the cycle are in the same strong component.  Let $v$ be the vertex of minimum $v.\pre$ on the cycle.  Let $\low$ values be their final (minimum) values.  Since $v$ is a follower, $v.\low < v.\pre$.  No descendant~$x$ of~$v$ can have $x.\low < v.\low$.  Let $w$ be the descendant of~$v$ whose $\low$ value first equals $v.\low$.  Then the low arc of~$w$ must be a cross or back arc entering a non-descendant of~$v$.  An induction on the number of tree arcs on the tree path from~$v$ to~$w$ proves that each of these tree arcs is the low arc of the vertex it exits.  Hence the path of low arcs from~$v$ leads to a non-descendant of~$v$, a contradiction.  
\end{proof}

\begin{corollary}
The set of low arcs defines a forest of in-trees rooted at the leaders of the strong components.
\end{corollary}

In-trees and out-trees can be used to verify that the strong components computed by Algorithm~T are indeed strongly connected.  The topological order of the components defined by the algorithm can be used to verify that no vertices in different components are mutually reachable.  Verifying both demonstrates that the strong components produced by the algorithm are correct.  The computations of the condensation and of the in-trees and out-trees can also be used to reduce the number of arcs needed to maintain reachability: Delete all arcs but those in an in-tree or an out-tree and all but one arc between each pair of strong components.  There may still be redundant arcs, i.e., arcs whose removal does not change reachability.  Finding a reachability-preserving subgraph with fewest arcs in a strongly connected graph is NP-complete.  Gibbons et al.~\cite{GKRST91} give two algorithms that run in $\OO(m+n\log n)$ time to find a reachability-preserving subgraph with no redundant arcs in a strongly connected graph.  Simon~\cite{Simon89} claims a linear-time algorithm, but his algorithm is not correct.  References \cite{AGU72,KRY95,Vetta01,ZNI03,BDK09} deal with the closely related notion of \emph{transitive reduction} and with approximation algorithms for the problem of finding a reachability-preserving subgraph with fewest arcs. 

\section{Implementation of Algorithm~T}\label{S:Tscc-implement}

Now we implement Algorithm~T.  We represent the strong components using a vertex pointer field $v.\ptr$ set equal to the leader of the strong component containing $v$. Thus two vertices~$u$ and~$v$ are in the same strong component if and only if $u.\ptr=v.\ptr$ when the algorithm finishes. We use the following four operations to maintain a stack~$S$:   $\stack()$, which returns an empty stack; $S.\push(v)$, which pushes $v$ onto stack~$S$; $S.\Top()$, which returns the top vertex on stack~$S$ if $S$ is non-empty, or returns $\NULL$ if $S$ is empty; and $S.\pop()$, which deletes and returns the top vertex on stack~$S$ if $S$ is non-empty, or throws an error if $S$ is empty.  We also allow the possibility of using $\pop$ without a return value, to merely delete the top element on the stack.  To avoid the need to test whether the stack~$F$ of followers is empty, we use a dummy \emph{guard} vertex $\guard$ with $\guard.\low = 0$.  We initialize~$F$ to contain just the guard, which is never popped.

We include one space-saving optimization.  There is no need to maintain $\pre$.  If $v$ is a vertex, $v.\low = v.\pre$ just after $v$ is previsited.  If $v.\low$ ever decreases, $v$ is not a leader.  Instead of maintaining preorder numbers, we maintain a bit $v.\lead$ for each vertex $v$, set to $\true$ when $v$ is previsited and $\false$ when $v.\low$ decreases.

The implementation uses the recursive implementation of depth-first exploration in Section~\ref{sub-recursive-dfs}.  It instantiates $\start$, $\unvisited$, $\previsit$, $\postvisit$, and $\retreat$ with appropriate computations and omits the other stubs.  It includes one other stub, $\leader(v)$, used in $\previsit$ to test whether $v$ is a leader.  We add this stub to avoid the need to store the leader bits separately, as we shall discuss.  Figure~\ref{F-Tscc} gives the instantiations of the stubs.

\begin{figure}
\centering
\parbox{2.2in}{
\begin{algorithm}[H]
\Fn{$\start(V)$}
{
    $\Time \gets 0$ \;
    $\guard.\low \gets 0$ \;
    $\F \gets \stack()$ \;
    $F.\push(\guard)$ \;
    \For{$v\in V$}
        {$v.\low \gets 0$}
}
\end{algorithm}
\vspace*{3pt}
\begin{algorithm}[H]
\Fn{$\unvisited(v)$}
{
    \Return $v.\low=0$ \;
}
\end{algorithm}
\vspace*{3pt}
\begin{algorithm}[H]
\Fn{$\leader(v)$}
{
    \Return $v.\lead$ \;
}
\end{algorithm}
}
\hspace*{-0.5in}
\parbox{2in}{
\begin{algorithm}[H]
\Fn{$\previsit(v)$}
{
    $\Time \gets \Time +1$ \;
    $v.\low \gets \Time$ \;
    $v.\lead \gets \true$ \;
}
\end{algorithm}
\vspace*{10pt}
\begin{algorithm}[H]
\Fn{$\retreat(v,a, w)$}
{
    \If{$w.\low < v.\low$}
    {
        $v.\low \gets w.\low$ \;
        $v.\lead \gets \false$ \;
    }
}
\end{algorithm}
}
\hspace*{-0.2in}
\parbox{2.8in}{
\begin{algorithm}[H]
\Fn{$\postvisit(v)$}
{
    \eIf{$leader(v)$}
    {
        \While{$F.\Top().\low \geq v.\low$}
        {
            $x \gets \F.\pop()$ \;
            $x.\ptr\gets v$ \;
            $x.\low \gets \infty$ \;
        }
        $v.\ptr\gets v$ \;
        $v.\low\gets \infty$ \;
    }
    {
        $F.\push(v)$ \;
    }
}
\end{algorithm}
}
\caption{Algorithm~T as an instantiation of the generic depth-first exploration algorithm.}\label{F-Tscc}
\end{figure}

It remains to implement the stack operations.  We represent the stack~$F$ by an endogenous singly linked list of its vertices.  Since a vertex is not on~$F$ when it is added to a strong component, the $ptr$ fields can do double duty: While~$v$ is on~$F$, $v.\ptr$ is the vertex below $v$ on~$F$; once $v$ is in a strong component, $v.\ptr$ is the leader of the strong component containing $v$.  Figure~\ref{F-stack} gives implementations of the stack operations.  In all our uses of $\pop$, the stack is guaranteed to be non-empty.  Therefore we omit the (otherwise required) test that the stack is non-empty.

\begin{figure}[t]
\parbox{1.6in}{
\begin{algorithm}[H]
\Fn{$\stack()$}
{
    \Return $\NULL$ \;
}
\end{algorithm}
}
\parbox{1.6in}{
\begin{algorithm}[H]
\Fn{$S.\push(v)$}
{
    $v.\ptr \gets S$ \;
    $S \gets v$ \;
}
\end{algorithm}
}
\parbox{1.6in}{
\begin{algorithm}[H]
\Fn{$S.\pop()$}
{
    $\Top \gets S$ \;
    $S \gets S.\ptr$ \;
    \Return $\Top$ \;
}
\end{algorithm}
}
\parbox{1.6in}{
\begin{algorithm}[H]
\Fn{$S.\Top()$}
{
    \Return $S$ \;
}
\end{algorithm}
}
\caption{Implementation of a stack as an endogenous list singly linked using the $\ptr$ field.}
\label{F-stack}
\end{figure}

\begin{figure}[t]
\centering
\parbox{2.1in}{
\begin{algorithm}[H]
\Fn{$\leader(v)$}
{
    \Return $v.\low \;\&\; 1=0$ \;
}
\end{algorithm}
}
\parbox{2.1in}{
\begin{algorithm}[H]
\Fn{$\previsit(v)$}
{
    $\Time \gets \Time +2$ \;
    $v.\low \gets \Time$ \;
}
\end{algorithm}
}
\parbox{2.2in}{
\begin{algorithm}[H]
\Fn{$\retreat(v,a,w)$}
{
    \If{$w.\low < v.\low$}{$v.\low \gets w.\low \;|\; 1$ \;}
}
\end{algorithm}
}
\caption{Alternative stubs that encode leader bits in \low\ values.}\label{F-Tscc1}
\end{figure}

If we are willing to indulge in a small programming trick, we need not store the leader bits separately but can encode them in the $\low$ values.  We store $\neg v.\lead$ in the low-order bit of $v.\low$, which we do by multiplying $v.\low$ by $2$ and adding $1$ if $v.\lead = \false$.  The first author and John Hopcroft used this encoding trick fifty years ago in their planarity-testing and triconnected components algorithms~\cite{HoTa73,HoTa74}.  The changes needed are in $\leader$, $\previsit$, and $\retreat$.  Figure~\ref{F-Tscc1} gives the modified versions.  In $\leader$, the operation ``$\&$" is bitwise ``and": If $b$ is an integer, ``$b \,\&\, 1$" returns the low-order bit of~$b$.  In $\retreat$, the operation ``$|$" is bitwise ``or": If $b$ is an integer, ``$ b\, |\, 1$" returns $b$ with the low-order bit set to~$1$.  

If the number of vertices in the input graph is a known quantity (which would \emph{not} for example be the case if the graph is given implicitly, say by a function that, given a vertex, computes the arcs out of the vertex), Algorithm~T can sometimes stop early, specifically when all of the vertices not yet found to be in a strong component are guaranteed to be in the same component.  This is true when all vertices have been previsited (as indicated by a count of previsited vertices) and the $\low$ value of the current vertex becomes equal to the $\low$ value of the start vertex of the current search.  When this happens, the algorithm can form a strong component containing all the vertices on the followers stack and all those on the current path, and then stop.  This is an example of ``stop when done," proposed by Kelly \cite{Kelly20a} as a way to speed up breadth-first search and similar computations.  Knuth (private communication, 2021) claims that stopping when done in Algorithm~T reduces the running time significantly in practice, especially on dense graphs.

The implementation we have presented is almost the same as Pearce's~\cite{Pearce16} version of Algorithm~T.  He maintains leader bits separately rather than encoding them in $\low$ values, and he uses large $\low$ values to represent the components, rather than using the stack pointers.

The implementation differs from the original 1972 version of Algorithm~T in several ways.  One is that all retreats update $\low$ values in the same way.  In the 1972 version, a retreat on a non-tree arc from~$v$ to~$w$ replaces $v.\low$ by $\min\{v.\low, w.\pre\}$ rather than~$\min\{v.\low, w.\low\}$.  The correctness proof of the 1972 version uses Lemma~\ref{L:follower-path} in place of lemma~\ref{L:low-path}.  Handling tree and non-tree arcs in the same way simplifies the code and eliminates the need to maintain preorder numbers.  Eve and Kurki-Suonio~\cite{EvKu77} and Duff and Reid~\cite{DuRe78} proposed this change independently.

Another difference is that the 1972 version adds \emph{all} the vertices to the stack, not just the followers, and it adds them when they are previsited, not when they are postvisited.  Nuutila and Soisalon-Soininen \cite{NuSo94} proposed stacking only the followers, and only when they are postvisited.  Their motivation was to avoid unnecessary additions to the stack. A more important effect of this change is that it allows us to implement the algorithm without recursion and without using additional space, by storing the followers stack and the recursion stack using the same space.  We describe how to do this in Section~\ref{S-non-recursive-dfs}.

A third difference is that the 1972 version does not set $\low$ values of vertices popped from the stack to~$\infty$.  Instead, when retreating on an arc~$a$ from~$v$ to~$w$, it sets $v.\low \gets \min\{v.\low, w.\low\}$ only if $w$ is on the stack.  Given that the 1972 algorithm adds vertices to the stack in preorder, this condition is equivalent to~$v$ and~$w$ being in the same strong component, as one can prove in the same way as the proof of Lemma~\ref{L:low-values-correct}.  Both Eve and Kurki-Suonio~\cite{EvKu77} and Duff and Reid~\cite{DuRe78} independently used the method of setting $v.\low$ to a high value when~$v$ is popped from the stack.

If the vertices are the integers from $1$ to $n$ (or from $0$ to $n-1$) and we know the value of $n$, we can use the $\low$ values instead of the vertex pointers to encode the components, as suggested by Knuth~\cite{Knuth21}: When a vertex $x$ is added to the component with leader $v$, we set $x.\low \gets v+n$.  This maintains the invariant that the $\low$ values of vertices already in components are larger than the $\low$ values of vertices not in components, and it allows us to compute the leader of the component containing a vertex $x$ by subtracting $n$ from $x$.  This frees $x.\ptr$ for another use once the component containing $x$ is found, and it saves one assignment per vertex.  This idea also works in Algorithm~C (Section~\ref{S-other-C}).

\section{Two Alternative Algorithms}\label{S-other}

Two other linear-time algorithms to find strong components are now known.  Both use depth-first search.  The first finds cycles and uses them to form larger and larger strongly connected vertex sets; the second does \emph{two} explorations, one forward, one backward.

\subsection{Strong components via cycle-finding}\label{S-other-C}

Algorithm~T is very efficient but not very intuitive.  There is a more intuitive algorithm that can be made efficient using some of the ideas we have developed, however.  We call a set of vertices \emph{strongly connected} if any pair of its vertices are mutually reachable by paths of vertices \emph{within} the set.  By Lemma~\ref{L:strong-component-paths} the strong components are strongly connected sets, and they are the maximal ones.  But in general they are not the only ones. The set of vertices on any cycle is a strongly connected set.  So is the union of any two non-disjoint strongly connected sets.  This suggests the following algorithm to find strong components: Start with each vertex in a singleton set, which is strongly connected by definition.  Find a cycle containing vertices in two or more of the existing strongly connected sets.  Unite the sets containing the vertices on this cycle.  Repeat until no cycle contains vertices in two or more strongly connected sets.  The final strongly connected sets are the strong components, because if two vertices not in the same set are mutually reachable, there is a cycle containing them, and the algorithm has not stopped.

An alternative view of this algorithm is as a process of cycle contraction. Find a cycle.  Contract all its vertices into a single vertex.  Repeat until every cycle, if any, is a loop.  (Contractions can produce loops, as well as parallel arcs, even if the original graph does not contain any.)  Each vertex~$v$ in the final contracted graph corresponds to a strong component in the original graph, whose vertices are those contracted to form~$v$.  We adopt the set union view  of the algorithm rather than the contraction view because it is more general: It allows us to handle the contractions implicitly.

We make the algorithm efficient by using a depth-first exploration to find cycles systematically.  We maintain a stack of strongly connected sets, initially empty.  When previsiting a vertex $v$, we initialize a strongly connected set $\{v\}$ and push it on the stack.  When traversing a non-tree arc~$a$ from~$v$ to~$w$, if $w$ is in a set on the stack, we pop all sets on the stack down to and including the set containing $w$, unite them, and push on the stack the set resulting from the unions.  When postvisiting a vertex $v$, the set~$S$ containing $v$ is on top of the stack.  If $v$ is the vertex in~$S$ with $v.\pre$ minimum, we pop~$S$ from the stack and declare it to be a strong component.

This is the \emph{cycle-finding} algorithm for strong components, which we call \emph{Algorithm~C}.  We shall prove the algorithm correct, but first we look at its early history, develop a linear-time implementation, and look at its later history.   

Sargent and Westerberg~\cite{SaWe64} proposed the contraction version of Algorithm~C (in 1964!), but they gave no efficiency analysis.  Instead of a forward exploration, they used a backward exploration, which produces the strong components in a topological order rather than a reverse topological order.  Purdom~\cite{Purdom70} later independently proposed the contraction version of Algorithm~C.  His implementation maintains the contracted vertices using a na\"{i}ve method, which results in an $\OO(n^2)$ time  bound.  Munro~\cite{Munro71b} improved the bound to $\OO(m+n\log n)$ using a ``relabel the smaller half" idea to maintain the contracted vertices.

None of these early versions of the algorithm runs in linear time.  To obtain a linear-time implementation, we consider the behavior of the algorithm in more detail.  If we view the algorithm as a process of set unions rather than cycle contractions, we see that the algorithm reduces the problem of finding strong components to a classical problem in data structures, the \emph{disjoint set union} or \emph{union-find} problem: Maintain a collection of disjoint sets subject to four operations: 

$\makeset(x)$: Create a singleton set $\{x\}$.  Element~$x$ must be in no existing set.

$\find(x)$: Return the set containing element~$x$.

$\unite(x, y)$: Given that~$x$ and~$y$ are in different sets, unite these sets.  This operation destroys the old sets containing~$x$ and~$y$.

$\set(x)$: Return a list of the elements in the set containing element~$x$.

Algorithm~C takes $\OO(m)$ time plus the time for $n$ $\makeset$ operations, $\OO(m)$ $\find$ and $\unite$ operations, and one $\set$ operation for each of the strong components.  One efficient way to do the $\set$ operations is to store each set as an endogenous circular singly linked list of its elements.  Catenating two such lists takes $\OO(1)$ time, totaling $\OO(n)$ over all the unions, and each $\set$ operation takes $\OO(1)$ time to return a pointer to the list of elements in the set.

The problem of implementing $\makeset$, $\find$, and $\unite$ operations has a fascinating history and results~\cite{Tarjan75}.  Several versions of a simple but hard to analyze data structure run in $\OO(m\alpha_{m/n}(n))$ time, where $\alpha_k(n)$ is a functional inverse of Ackermann's function (see, e.g., \cite{ATGRZ14}). The~$\alpha$ function grows extremely slowly and for all practical purposes is constant.     

The $\OO(m\alpha_{m/n}(n))$ bound is tight for general instances of the disjoint set union problem~\cite{FrSa89}, but those that occur in the cycle-finding algorithm are \emph{not} completely general.  They have the \emph{incremental-tree} property: Each $\unite$ unites a vertex with its parent in the depth-first forest.  Such instances can be solved in $\OO(m)$ time by using a more complicated data structure~\cite{GaTa85}.

Thus we can implement the cycle-finding algorithm to run in almost-linear or linear time by using an appropriate solution to the disjoint set union problem.  But the set union instances that occur in the algorithm are even more restricted: The unite operations obey a stack discipline, and any $\find$ of a set that is on the stack but not on top results in one or more corresponding $\unite$ operations.  By exploiting this and using our previous results on depth-first search, we can implement the cycle-finding algorithm without the need for any additional data structures.  The result is a strong components algorithm closely resembling Algorithm~T.

We number the vertices consecutively from $1$ in preorder, as in Section~\ref{S-dfs}.  We define the \emph{leader} of a strongly connected set to be the vertex $v$ in it with $v.\pre$ minimum.  Instead of maintaining a stack of sets, we maintain a stack $\LL$ of their leaders.  We run Algorithm~C as follows: When previsiting a vertex~$v$, set $v.\pre$ equal to the next available natural number, make $\{v\}$ a set with leader $v$, and push $v$ on $\LL$.  While $\LL.\Top().\pre > a.\tip.\pre$,
pop the top vertex from~$\LL$ and unite its set with that of the new top vertex on~$\LL$, which becomes the leader of the new set. 
Repeat until $\LL.\Top().\pre \leq a.\tip.\pre$.  As we shall show, if~$S$ is the set containing $\LL.\Top()$ after all the unites caused by the traversal of~$a$, $S$ is a strongly connected set.  When postvisiting a vertex~$v$, if $v$ is on top of~$\LL$, pop $v$, set $w.\pre$ to~$\infty$ for every vertex $w$ in the set with leader $v$, and declare this set to be a strong component.

This is a \emph{simple} implementation of Algorithm~C, and it runs in linear time.

\begin{theorem}
Algorithm~C is correct.
\end{theorem}

\begin{proof}
A straightforward induction verifies that the algorithm maintains the following invariants: 
\begin{enumerate}
\item[(i)] The vertices in sets are exactly the previsited vertices. \item[(ii)] The leaders on $\LL$ are a subset of the vertices on the current path, in the same order on $\LL$ as on the current path.  (The topmost leader on $\LL$ is the last leader on the current path.) \item[(iii)] The leader of a set is its vertex $v$ with $v.\pre$ minimum. \item[(iv)] Every vertex on the current path is in a set whose leader is on $\LL$. \item[(v)] If $v$ is on $\LL$, and $w$ is just above $v$ on $\LL$, the vertices in the set with leader $v$ are exactly the vertices~$x$ with $v.\pre \leq x.\pre < w.\pre$; if $v$ is on top of~$\LL$, the vertices in its set are exactly the vertices~$x$ with $v.\pre \leq x.\pre < \infty$. \item[(vi)] Except possibly in the middle of the traversal of a non-tree arc, the sets are strongly connected. \item[(vii)] If $v$ is on top of~$\LL$ when $v$ is postvisited, its set is a strong component.
\end{enumerate}

The only invariants whose verification requires a bit of argument are (vi) and (vii).  To verify (vi), suppose that the traversal of a non-tree arc~$a$ from~$x$ to~$y$ unites two or more sets whose leaders are on $\LL$.  Let $v$ be the current vertex, and let $u$ be minimum of the leaders of the sets united by the traversal.  Assuming all the invariants hold before the traversal, by (v) $y$ is in the set whose leader is $u$.  There is a cycle consisting of the arc~$a$ followed by the path of tree arcs from $u$ to~$v$.  This cycle contains the leaders of all the sets united by the traversal, by (ii).  Hence uniting all the sets preserves (vi).

To verify (vii), suppose the invariants hold just before the postvisit of a vertex $v$ that is on top of~$\LL$.  Let~$a$ be an arc from~$x$ to~$y$ with~$x$ in the set~$S$ with leader $v$ when $v$ is postvisited.  We claim that either~$y$ is in~$S$ or~$y$ is in a previously declared strong component.  To prove the claim, consider such an arc~$a$.  If $y.\pre = \infty$ just before $v$ is prevsited, the claim holds.  Suppose $y.\pre$ is finite just before $v$ is postvisited.  If $y.\pre \geq v.\pre$, then~$y$ is in~$S$ by (v).  If $y.\pre < v.\pre$, $a$ is a non-tree arc and $y.\pre$ was finite when~$a$ was traversed.  After this traversal,~$x$ and~$y$ are in the same set.  Hence the claim holds.
\end{proof}

The first, more general implementation of Algorithm~C is also correct, because its behavior is isomorphic to that of the implementation with set leaders.

We obtain an implementation of Algorithm~C from the generic implementation of depth-first exploration in Section~\ref{sub-recursive-dfs} by instantiating stubs $\start$, $\unvisited$, $\previsit$, $\postvisit$, and $\nontreetraverse$, and omitting the other stubs.  Figure~\ref{F-ccscc} gives the instantiations.  As we mentioned previously, one can efficiently maintain each strongly connected set explicitly, by storing its vertices in a singly linked circular list.  This takes an extra pointer per vertex, however.  Instead, we use the same idea as in Algorithm~T: Push vertices found to be followers onto a stack $\F$.  (Each such vertex moves from~$L$ to~$F$.)  When a postorder visit pops the leader of a strong component from $\LL$, pop its followers from $\F$ to form the component.  As in Algorithm~T, we represent both stacks as endogenous singly linked lists using $\ptr$ fields to store the pointers.  Figure~\ref{F-stack} gives an implementation of the stack operations.  We use a guard vertex $\guard$ on the bottom of~$F$ to avoid testing for an empty stack.  Stack~$L$ does not need a guard.  Also as in Algorithm~T, we use the $\ptr$ fields to represent the strong components, by setting $v.\ptr$ equal to the leader of the component containing $v$ when the strong component containing $v$ is declared.  An alternative as in Algorithm~T is to represent the strong components by singly linked lists of their vertices using the $\ptr$ fields to store the pointers.  

\begin{figure}[t]
\centering
\hspace*{-0.2in}
\parbox{1.8in}{
\begin{algorithm}[H]
\Fn{$\start(V)$}
{
    $\Time \gets 0$ \;
    $\LL \gets \stack()$ \;
    $\F\gets \stack()$ \;
    $\guard.\pre \gets 0$ \;
    $F.\push(guard)$ \;
    \SPC{0.5}
    \For{$v\in V$}
    {
        $v.\pre \gets 0$ \;
    }
}
\end{algorithm}
\vspace*{5pt}
\begin{algorithm}[H]
\Fn{$\unvisited(v)$}
{
    \Return $v.\pre=0$ \;
}
\end{algorithm}
}
\hspace*{-0.3in}
\parbox{2.55in}{
\begin{algorithm}[H]
\Fn{$\previsit(v)$}
{
    $\Time \gets \Time +1$ \;
    $v.\pre \gets \Time$ \;
    $\LL.\push(v)$ \;
}
\end{algorithm}
\vspace*{5pt}
\begin{algorithm}[H]
\Fn{$\nontreetraverse(v,a,w)$}
{
    \While{$w.\pre < \LL.\Top().\pre$}
    {
        $F.\push(\LL.pop())$ \;
    }
}
\end{algorithm}
}
\hspace*{-0.35in}
\parbox{2.75in}{
\begin{algorithm}[H]
\Fn{$\postvisit(v)$}
{
    \If{$v = \LL.\Top()$}
    {
        \While{$ v.\pre < F.\Top().\pre$}
        {
            $x \gets F.\pop()$ \;
            $x.\ptr \gets v$ \;
            $x.\pre \gets \infty$ \;
        }
        \SPC{0.5}
        $\LL.\pop()$ \;
        $v.\ptr \gets v$ \;
        $v.\pre \gets \infty$ \;
    }
}
\end{algorithm}
}
\caption{Algorithm~C as an instantiation of the generic depth-first search algorithm.}\label{F-ccscc}
\end{figure}

The proof that the stack mechanism used in the pseudocode correctly finds the components and correctly sets $\pre$ values to~$\infty$ is virtually identical to that of Lemma~\ref{L:component-finding-correct}.  Just like Algorithm~T, Algorithm~C can stop early if all vertices have been previsited and the vertices not yet found to be in a strong component are guaranteed to be in the same component, specifically when all vertices have been previsited and the leaders stack contains only the start vertex of the search (in addition to the guard vertex): The last strong component contains all the vertices on the followers stack and the start vertex.  

Dijkstra~\cite{Dijkstra76} was apparently the first to publish a linear-time implementation of Algorithm~C. (An earlier version of his algorithm \cite{Dijkstra73,Dijkstra82} runs in $\OO(m + n^2)$ time.)  Later and independently Cheriyan and Mehlhorn~\cite{ChMe96} and Gabow~\cite{Gabow00} presented linear-time implementations of Algorithm~C.  (Cheriyan and Mehlhorn cited Dijstra's earlier, less efficient implementation, but were evidently unaware of his improved version.)

One difficulty in understanding Dijkstra's version is that he uses extendible arrays instead of stacks as his data structures (even though his arrays function as stacks), and he does not mention depth-first search, even though his version does a depth-first exploration, implemented non-recursively but inefficiently.  The implementations of Cheriyan and Mehlhorn and of Gabow differ only slightly from the one we have presented.  They both use a recursive depth-first search.  Dijkstra uses four stacks, two of which are the stacks of leaders and followers, $L$ and~$F$.  He uses two additional stacks to mimic enough of the effect of a depth-first search to implement Algorithm~C.  Finally, he uses a number per vertex.  One of the two stacks he uses to mimic a depth-first search is a stack~$S$ of vertices.  The other is a stack of positions in~$S$, one per vertex on~$L$.  In the middle of a search from start vertex $s$, $S$ contains $s$ on the bottom and the tips of all untraversed arcs from previsited vertices.  The next vertex to be previsited is the topmost unvisited vertex on~$S$.  When a vertex is previsited, the tips of its outgoing arcs are pushed onto~$S$ if they are unvisited.  Stack~$S$ can contain many copies of a vertex at the same time, one per entering arc from a visited vertex.  The total space overhead of Dijkstra's implementation is thus $m + 4n$, $m + n$ more than needed (as we shall discuss).  His implementation also does more memory accesses than needed, by constant factors in~$m$ and~$n$.

All three versions push \emph{all} vertices on~$F$, not just the followers.  Thus a vertex can be on both~$L$ and~$F$ at the same time.  Dijkstra's and Gabow's versions number the vertices according to their position on~$F$, not consecutively in preorder, and they store the vertex numbers, not the vertices themselves, on~$L$.  They also represent components using the vertex numbers, by assigning the vertices in each new component a component number larger than $n$ (and hence serving as $\infty$ in the leader test.)  Cheriyan and Mehlhorn's version does not set the numbers of vertices in complete components to~$\infty$ but instead maintains for each vertex a bit indicating whether it is in a declared strong component or not.  Except for Dijkstra's implementation of the exploration, all these differences are minor, although they do affect performance.  Dijkstra's version is the least space-efficient, requiring $m + 4n$ words of storage, and also the least time-efficient.  If the other two algorithms use the non-recursive implementation of depth-first exploration in Section~\ref{S-non-A}, they need $4n$ words of storage.  Our version reduces this to~$3n$ by using the same pointers to represent~$L$ and~$F$.   

Algorithm~C produces the strong components in the same order as Algorithm~T.  It produces the vertices in each component in preorder rather than postorder, as does the 1972 version of Algorithm~T but not the one we have presented here. Algorithm~T and the Algorithm~C are quite similar in the number of steps they do. To make a finer comparison between Algorithm~T and Algorithm~C, we count memory accesses, as suggested by Knuth~\cite[p.~460]{Knuth93}. (See Section~\ref{S-comp-VA}.)
We obtain upper bounds of~$7n$ and $9n$ for the number of memory accesses done by Algorithm~T and Algorithm~C, respectively, in addition to those required just to do the depth-first exploration.  Including the accesses required by an efficient non-recursive implementation of the depth-first exploration, the total number of memory accesses is at most $3m+16n$ for Algorithm~T, at most $3m+18n$ for Algorithm~C. See Section~\ref{S-non-recursive-dfs}. This may suggest that Algorithm~T is slightly more efficient, especially if $m$ is not much larger than~$n$, but this needs to be checked experimentally. Algorithm~T can also use the $\ptr$ fields to store the recursion stack for the depth-first search, as we discuss in Section~\ref{S-non-A}.  Thus it needs only~$2n$ extra words of storage, versus the~$3n$ of our implementation of Algorithm~C.

Algorithm~C superficially resembles another classical graph algorithm: Edmonds's \cite{Edmonds65} blossom-shrinking algorithm for finding an augmenting path in a graph with a matching, if the latter uses depth-first search to find an augmenting path~\cite{Tarjan83}.  Kameda and Munro~\cite{KaMu74} gave an implementation of Edmonds's algorithm much like the linear-time implementation of Algorithm~C.  Their implementation is not correct, however: In Edmonds's algorithm the graph is undirected, and there is a dynamic parity constraint imposed by the matching and the set of blossoms.  One can implement Edmonds's algorithm using a general disjoint set union data structure, resulting in a running time of $\OO(m\alpha_{m/n}(n))$.  The set union instances have the incremental-tree property, so a data structure for such instances can be used, resulting in an $\OO(m)$ time bound.  But the unite and find operations do \emph{not} obey a stack discipline, so a simple solution using a stack, as proposed by Kameda and Munro, does not work.  A much more recent paper~\cite{YuSh14} gives the same algorithm, with the same error.

There is another application of disjoint set union in which the stack-based implementation we have described here \emph{does} work.  This is an algorithm for verifying the correctness of a priority queue~\cite{FiMe99}.  The authors claim a linear time bound based on the incremental-tree set union method, but the method used in Algorithm~C also works and is much simpler. 

\subsection{Strong components via bidirectional search}

There is an even simpler way to find strong components using depth-first search, one that uses almost no data structures: Do a depth-first exploration.  Build a list of the vertices in postorder.  Then do a \emph{backward} exploration, choosing start vertices for successive searches in the reverse of the postorder generated by the forward search.  Each start vertex is the leader of its strong component, and the backward search from it visits exactly the vertices in its strong component.  A backward exploration is equivalent to a forward exploration on the \emph{reverse} of the problem graph.  The reverse is formed by reversing every arc (swapping its ends).  This is the \emph{bidirectional} algorithm for strong components, \emph{Algorithm~B}.  The forward exploration can stop once all vertices have been previsited; the backward exploration can stop once all vertices have been visited. 

In Algorithm~B, the backward search does not need to be depth-first; \emph{any} search order, such as breadth-first, will do.  Kosaraju discovered Algorithm~B but never published it (see \cite{AHU83}); Sharir~\cite{Sharir81} discovered it independently.  In spite of its brevity, its correctness relies on non-trivial properties of depth-first search, specifically Lemma~\ref{L:strong-component-leaders} and Theorem~\ref{T:component-top-order}:

\begin{theorem}\label{T-bidirectional}
Algorithm~B is correct.
\end{theorem}

\begin{proof}
The proof is by induction on the number of strong components found.  For each vertex $v$, let $v.\post$ be the time $v$ is postvisited during the forward exploration.  Suppose the backward exploration correctly finds components before a backward search starting from $s$ in strong component~$S$.  Then the strong components containing all vertices $v$ with $v.\post > s.\post$ have been found before the search from $s$ starts.  When the search from $s$ starts, $S$ is not yet found, and $s$ is the vertex not yet in a found component with $s.\post$ maximum, so by Lemma~\ref{L:strong-component-leaders} $s$ is the leader of~$S$, and no vertices in~$S$ have been visited by the backward search.  Hence the backward search from $s$ visits all vertices in~$S$.  By Theorem~\ref{T:component-top-order}, if there is an arc from~$x$ to~$y$ with~$y$ but not~$x$ in~$S$, the component containing~$x$ is found before the search from $s$ begins.  Hence the search from $s$ visits only the vertices in~$S$.
\end{proof}

Algorithm~B runs in $\OO(m)$ time.  We omit the straightforward implementation.  Although Algorithm~B is conceptually simple, it requires two searches, not one.  Assuming the graph representation is read-only, Algorithm B requires twice as much space as Algorithms T and C for the graph representation, since both forward and backward arc lists for each vertex must be represented: $4m+2n$ words of storage versus $2m+n$ (two pointers per vertex to forward and backward incidence lists plus two ends of each arc plus two $\next$ pointers per arc).  The extra space needed for finding strong components is $2n$, the same as for Algorithm~T. If an efficient non-recursive implementation (see Section~\ref{S-non-recursive-dfs}) is used for the forward search and quick search (see the appendix) is used for the backward search, Algorithm~B takes $6m +14n$ memory accesses, versus $3m+16n$ for Algorithm~T, making Algorithm~T noticeably more efficient.  If the graph representation can be modified and then restored, the space needed by Algorithm B can be reduced, but at the cost of significant additional time to construct the incoming arc lists after the forward search and then restore the outgoing arc lists after the backward search.

\section{Non-Recursive Depth-First Search}\label{S-non-recursive-dfs}

\begin{figure}[t]
\centering
\parbox{2in}{
\begin{algorithm}[H]
\Fn{$\dfexp(V)$}
{
    $\start(V)$ \;
    $P \gets \stack()$ \;
    \For{$s\in V$}
    {\If{$\unvisited(s)$}{$\dfs(s)$ \;}}
}
\end{algorithm}
}
\hspace*{0.2in}
\parbox{2.8in}{
\begin{algorithm}[H]
\Fn{$\dfs(s)$}
{
    $v \gets s$ \;
    $\previsit(v)$ \;
    $a \gets v.\first$ \;
    \While{$\true$}
    {
        \eIf{$a \neq \NULL$}
        {
            $w \gets a.\tip$ \;
            $\Advance(v, a, w)$ \;
            \If{$\unvisited(w)$}
            {
                $\treeadvance(v, a, w)$ \;
                $\previsit(w)$  \;    
                $\FORWARD$ \;
                {\bf continue} \;
            }            
            $\nontreetraverse(v, a, w)$ \;
        }
        {
            $\postvisit(v)$ \;
            \If{$v = s$}{\Return \;}
            $\BACKWARD$ \;
            $\treeretreat(v,a, w)$ \;
        }
        $\retreat(v, a, w)$ \;                 
        $a \gets a.\next$ \;
    }
}
\end{algorithm}
}
\caption{A non-recursive implementation of a generic depth-first exploration. $\FORWARD$ and $\BACKWARD$ perform the necessary updates of $v$, $a$ and $w$ using a stack~$P$ of vertices or arcs.}
\label{F-dfs-non}
\end{figure}

Although the recursive view of depth-first search presented in Section~\ref{S-dfs} is compact and elegant, a non-recursive implementation allows the programmer to choose the representation of the recursion stack, possibly saving time and space.  This may well be critical in handling very large graphs.  For example, SciPy~\cite{Leslie}, a widely used library for scientific computing, uses a non-recursive implementation, which saves space and makes the algorithm significantly faster.\footnote{The non-recursive SciPy implementation differs from the non-recursive implementations suggested here and can possibly be improved.}  

In this section we study two alternative non-recursive implementations of depth-first search. They both share the
high-level code given in Figure~\ref{F-dfs-non}. They differ in the implementation of the two short code segments \FORWARD\ and \BACKWARD\ responsible for explicitly handling a stack $P$ that replaces the recursion stack. In Implementation~V, for \emph{\underline{V}ertex stack}, $P$ is a stack of vertices. In Implementation~A, for \emph{\underline{A}rc stack}, $P$ is a stack of arcs. The stack~$P$ represents the current path.

As in the recursive implementation of Figure~\ref{F-dfs}, $s$ is the start vertex of the current search, $v$ is the current vertex, and~$a$ is the currently explored arc. The non-recursive depth-first search code of Figure~\ref{F-dfs-non} is structured as a while loop that iterates over the current arc~$a$. There are three cases, depending on whether~$a$ is a tree arc, a non-tree arc, or $\NULL$. The implementation again includes stubs for the various critical steps of the algorithm: $\previsit$, $\postvisit$ etc. These stubs can be viewed either as function calls or as short code segments that are replaced inline. As in the recursive formulation, $\start(V)$ must mark all vertices unvisited, and $\previsit(v)$ must mark~$v$ visited.

\subsection{Implementation~V}\label{S-non-V}

\begin{figure}[t]
\SetKwProg{Fn}{}{$\;\equiv$}{}
\centering
\parbox{2in}{
\begin{algorithm}[H]
\Fn{$\FORWARD$}
{
    $v.\arc \gets a$ \;
    $P.\push(v)$ \;
    $v \gets w$ \;
    $a \gets v.\first$ \;
}
\end{algorithm}
\vspace*{0.3in}
\begin{algorithm}[H]
\Fn{$\BACKWARD$}
{
    $w \gets v$ \;
    $v \gets P.\pop()$ \;
    $a \gets v.\arc$ \;
}
\end{algorithm}
\vspace*{0.3in}
\begin{center}
(a)
\end{center}
}
\hspace*{0.2in}
\parbox{2in}{
\begin{algorithm}[H]
\Fn{$\FORWARD$}
{
    $v.\arc \gets a$ \;
    \If{$P.\Top()\ne v$}{$P.\push(v)$ \;}
    $v \gets w$ \;
    $a \gets v.\first$ \;
}
\end{algorithm}
\vspace*{0.2in}
\begin{algorithm}[H]
\Fn{$\BACKWARD$}
{
    $w \gets v$ \;
    \If{$P.\Top()=v$}{$P.\pop()$}
    $v \gets P.\Top()$ \;
    $a \gets v.\arc$ \;
}
\end{algorithm}
\begin{center}
(b)
\end{center}
}
\caption{Implementation V. (a) A basic implementation using a stack of vertices. (b) A possibly optimized version that does fewer push and pop operations at the cost of a few more tests.}
\label{F-dfs-non-V}
\end{figure}

In Implementation~V, the stack~$P$ holds the vertices on the current path, not including the current vertex~$v$. In addition, for each vertex~$u$ on~$P$, the field $u.\arc$ holds the arc out of~$u$ that is on the current path.

The corresponding implementations of \FORWARD\ and \BACKWARD\ are given in Figure~\ref{F-dfs-non-V}(a). (The~$\equiv$ following $\FORWARD$ and $\BACKWARD$ indicates that the  code segments should be replaced inline and should not be viewed as function calls.) 

\FORWARD\ corresponds to advancing from~$v$ to~$w$ along arc~$a$: $v.\arc$ is set to~$a$, $v$ is pushed onto the stack~$P$, $v$ becomes $w$, and the current arc $a$ becomes $v.\first$, the first outgoing arc of~$v$.

\BACKWARD\ corresponds to retreating from~$v$. First $w$ is set to~$v$, the vertex retreated from. The new vertex~$v$ is obtained by popping the stack~$P$. The current arc $a$ becomes $v.\arc$ but would soon be updated by the assignment $a\gets a.\next$ at the end of the while loop. (Vertices~$v$ and~$w$ and arc~$a$ are used by the stub $\treeretreat(v,a,w)$ that immediately follows \BACKWARD.)

A slightly optimized version of Implementation~V is given in Figure~\ref{F-dfs-non-V}(b). In the version of Figure~\ref{F-dfs-non-V}(a), when a retreat from~$w_1$ to~$v$ is followed by an advance from $v$ to~$w_2$, $v$ is first popped from~$P$ and then immediately pushed back onto~$P$. In Figure~\ref{F-dfs-non-V}(b) this is avoided at the cost of two simple tests that only involve values that can be stored in registers. In this variant, each non-leaf vertex is pushed and then popped from the stack~$P$ exactly once.

Implementation~V requires two extra words of storage per vertex on $P$, one to represent the stack, the other to indicate the current tree arc out of the vertex, for a total of at most $2n$ words. This representation is redundant: The start vertex~$s$ plus the tips of the arcs on the current path are exactly the vertices on the current path. This observation leads to our second non-recursive implementation, Implementation~A, considered below.

If the graph representation is not read-only but can be modified during the search and then restored, one can encode the current arc pointers in the graph representation, reducing the extra storage to that needed for the vertex stack.  We use $v.\first$ to point to the current arc of~$v$.  When stepping through the outgoing arc list of~$v$, we reverse the $\next$ pointers, so that when we get to the end of the list, $v.\first$ actually points to the \emph{last} arc on the list.  During the postorder visit to~$v$, we reverse all the $\next$ pointers again, restoring the original list order. If we do not care about the order, we do not need to do the reversal.  An alternative is to make each list of outgoing arcs circular, with a bit indicating the first arc on the list.  Then stepping through the list requires merely changing the $\first$ pointer, and the list stays in the same order throughout the process.

In Appendix~\ref{A-quick} we discuss a search method called \emph{Quick Search} that uses a vertex stack, as does Implementation~V, but no other data structure. Quick search is an efficient generic search method, but it does not have the specific properties of depth-first search needed in several interesting applications, including that of finding strong components in linear time.

\subsection{Implementation~A}\label{S-non-A}

In Implementation A, the stack~$P$ holds the arcs on the current path. The resulting representation is non-redundant and the $u.\arc$ fields are no longer needed, saving space and also some time.

\begin{figure}[t]
\SetKwProg{Fn}{}{$\;\equiv$}{}
\centering
\parbox{2in}{
\begin{algorithm}[H]
\Fn{$\FORWARD$}
{
    $P.\push(a)$ \;
    $v \gets w$ \;
    $a \gets v.\first$ \;
}
\end{algorithm}
\vspace*{0.3in}
\begin{algorithm}[H]
\Fn{$\BACKWARD$}
{
    $a \gets P.\pop()$ \;
    $w \gets v$ \;
    \eIf{$P.\Top()=\NULL$}
    {$v\gets s$\;}
    {$v\gets P.\Top().\tip$\;}
}
\end{algorithm}
\begin{center}
(a)
\end{center}
}
\hspace*{0.2in}
\parbox{2in}{
\begin{algorithm}[H]
\Fn{$\FORWARD$}
{
    $w.\ptr \gets P$ \;
    $P \gets a$ \;
    $v \gets w$ \;
    $a \gets v.\first$ \;
}
\end{algorithm}
\vspace*{0.2in}
\begin{algorithm}[H]
\Fn{$\BACKWARD$}
{
    $a \gets P$ \;
    $P \gets v.\ptr$ \;
    $w \gets v$ \;
    \eIf{$P=\NULL$}
    {$v\gets s$\;}
    {$v\gets P.\tip$\;}

}
\end{algorithm}
\begin{center}
(b)
\end{center}
}
\caption{Implementation A. (a) An implementation using an abstract stack of arcs. (b) With a direct implementation of the stack as a singly linked list with indirection via vertex pointers. The assignment $P\gets \stack()$ in $\dfe(V)$ should be replaced by $P\gets\NULL$.}
\label{F-dfs-non-A}
\end{figure}

\begin{figure}[t]
\parbox{1.6in}{
\begin{algorithm}[H]
\Fn{$\stack()$}
{
    \Return $\NULL$ \;
}
\end{algorithm}
}
\hspace*{-0.2in}
\parbox{1.7in}{
\begin{algorithm}[H]
\Fn{$P.\push(a)$}
{
    $a.\tip.\ptr \gets P$ \;
    $P \gets a$ \;
}
\end{algorithm}
}
\hspace*{-0.2in}
\parbox{1.7in}{
\begin{algorithm}[H]
\Fn{$P.\pop()$}
{
    $\Top \gets P$ \;
    $P \gets P.\tip.\ptr$ \;
    \Return $\Top$ \;
}
\end{algorithm}
}
\parbox{1.6in}{
\begin{algorithm}[H]
\Fn{$P.\Top()$}
{
    \Return $P$ \;
}
\end{algorithm}
}
\caption{Implementation of an arc stack as a linked list with one level of indirection.}
\label{F-stack-indirect}
\end{figure}

The corresponding implementations of \FORWARD\ and \BACKWARD\ are given in Figure~\ref{F-dfs-non-A}. The code in Figure~\ref{F-dfs-non-A}(a) uses a generic stack, e.g., the one given in Figure~\ref{F-stack}. This is quite wasteful in our case, as each arc $a$ on the stack would need to have a pointer $a.\ptr$ to the arc just below it on the stack. (Alternatively, the stack can be implemented as an exogenous singly-linked list, but this is also not optimal.)

An elegant way of maintaining the stack, proposed by Tarjan~\cite{Tarjan83b}, is to use indirection and store the pointer to the arc below $a$ on the stack in $a.\tip.\ptr$.
Equivalently, if $u$ is a vertex on the current path, $u.\ptr$ is the tree arc entering the \emph{parent} of~$u$, if~$u$ has a grandparent. The corresponding implementation of the arc stack is given in Figure~\ref{F-stack-indirect}.
Just before postvisiting a vertex $v$, Implementation~A pops from $P$ the tree arc entering $v$. This leaves $v.\ptr$ free for other uses after~$v$ is postvisited.
A complete version of Implementation~$A$, with a direct implementation of the arc stack~$P$ using this approach is given in Figure~\ref{F-dfs-non-A}(b). We treat Figure~\ref{F-dfs-non-A}(b) as the canonical version of Implementation~A.

\subsection{Comparing Implementations V and A}\label{S-comp-VA}

Implementation V is more straightforward, but Implementation A is slightly more efficient in both space and time.  We expect both implementations to be significantly faster than the recursive implementation.

Implementation~V needs at most $2n$ extra words of storage, to hold $u.\arc$ and $u.\ptr$ for each vertex on the active path. Implementation~A needs at most $n$ extra words to hold $u.\ptr$. In some cases, e.g., Algorithm~T for computing strong components, the space used by the $u.\ptr$ pointers can be reused for other purposes. In such cases Implementation~A needs no extra storage at all.

To make a finer comparison of Implementations~V and~A, as well as Algorithms~T,B, and~C, we count memory accesses (reads and writes) using a model of Knuth~\cite{Knuth93}. In this model, an algorithm is allowed to hold a small constant number of values in registers, access to which is free. Computations on values stored in registers is also free. Each read or write of a word from/to main memory counts as one memory access. This is quite a simplistic model that assumes that memory access is the dominating factor in the running time of an algorithm. Bounds derived in this model should thus be taken with a large grain of salt: They do not take computation into account and they do not consider the memory hierarchy, e.g., the existence of a cache, of the machine used. They are only useful for a rough comparisons for implementations or algorithms with the same asymptotic running time.

We state bounds for Implementations~V and~A in terms of~$m$, the number of arcs, and~$n$, the number of vertices, and ignore additive constants and dependencies on the number of search start vertices and the number of strong components. Assuming that $v$, $a$, and $w$ are stored in registers, Implementation~V does at most $3m+10n$ memory accesses, while Implementation~A does at most $3m+9n$ memory accesses, a small `moral' victory for Implementation~A. The $3m$ in both these bounds comes from the fact that for each arc~$a$ we read $a.\tip$, then $a.\tip.\pre$ (or $a.\tip.\low$), and finally $a.\next$. Determining the coefficients of~$n$ in these bounds requires a more careful examination of the two implementations. We omit the relatively straightforward, but somewhat tedious, derivations of the exact bounds for Implementations~V and~A and for Algorithms~T, C, and B.

\subsection{Non-recursive implementations of strong component algorithms}

Implementation A uses a pointer per vertex to store the recursion stack.  If we implement Algorithm~T using Implementation A, we can avoid this extra space overhead by letting each vertex pointer $v.\ptr$ do triple duty: Before $v$ is postvisited $v.\ptr$ is a link in the arc stack~$P$, after $v$ is postvisited it is a link in the followers stack~$F$, and after $v$ is added to a strong component it points to the leader of this component.  The total space needed by this implementation is $2n$ words, one integer between $0$ and $2n+1$ per vertex (if we represent~$\infty$ by $2n+1$), and one pointer capable of pointing to a vertex or arc per vertex.  The total number of memory accesses (assuming that $v.\low$ is stored in a register when $v$ is the current vertex) is at most $3m + 16n$, only $7n$ more than that needed just to do the search. 

Pearce's version~\cite{Pearce16} of Algorithm~T assumes the graph representation is read-only and uses Implementation~V instead of Implementation~A to do the search. His implementation requires $3n$ words and $n$ bits of extra storage: Only one of the two pointers per vertex needed for the search can be shared with the one pointer per vertex needed by Algorithm~T, and he uses an extra leader bit per vertex, instead of encoding these bits in the $\low$ values.

An alternative way to implement Algorithm~T non-recursively using Implementation~A is to store the recursion stack~$P$ and the followers stack~$F$ at the two ends of a single array of~$n$ positions.  This alternative needs to use the $\low$ values to represent the strong components as they are found.  The extra space needed remains~$2n$ words.

If we implement Algorithm~C using Implementation~V or~A to do the search, the space overhead is $n$ words more than needed by Algorithm~T: $4n$ versus $3n$ with Implementation~V, $3n$ versus~$2n$ with implementation A.  This is because Algorithm~C adds vertices to the leaders stack when they are previsited, which prevents the representations of the recursion stack and the leaders stack from sharing space.  The number of memory accesses is also larger by $2n$, at most $3m+18n$ with Implementation A.

If the graph representation is read-only, Algorithm B requires twice as much space as Algorithms T and C for the graph representation, since both forward and backward arc lists for each vertex must be represented.  The extra space needed by Algorithm B for finding strong components is $2n$, the same as Algorithm~ T: one pointer per vertex to implement the graph explorations and one pointer per vertex used for two purposes: to represent a list of the vertices in reverse postorder computed during the forward exploration and to store the mapping from vertices to their component leaders computed during the backward exploration.  These computations require $2n$ writes, $n$ during each search.  The backward exploration can reuse the $\visited$ bits set by the forward exploration and does not need to reinitialize them, saving $n$ reads and $n$ writes.  If Implementation~A is used for the forward search and quick search (see the appendix) is used for the backward search, Algorithm~B takes $6m +14n$ memory accesses, versus $3m+16n$ for Algorithm~T.

If graph representation is allowed to be changed during the search and then restored, Algorithm B can use the same graph representation as Algorithms T and C, but it must convert the representation to incoming arc lists after the forward search and restore it after the backward search.  This requires significant additional time.

If the graph representation is changeable, Implementation~A can be modified to use \emph{no} extra space, as long as each $\first$ and $\next$ pointer is capable of indicating a vertex or an arc, and vertices and arcs can be distinguished~\cite{Tarjan83b}.  This gives a way to implement all three strong components algorithms (and indeed \emph{any} depth-first search application) without using extra space for the graph exploration.

\section{Final Remarks}

We have presented and proved correct three linear-time algorithms for finding strong components: Algorithm~T, the original, Algorithm~C, based on the intuitive idea of finding and contracting cycles, and Algorithm~B, a simple algorithm that does both a forward and a backward search.  Though simple, Algorithm~B needs more time and needs both outgoing and incoming arc lists.  Algorithms~T and C have similar running times, but T uses noticeably less space and slightly less time. 
Finding strong components is a basic problem in algorithmic graph theory, and all three of the algorithms we have presented have found their way into textbooks.  Aho et al. \cite{AHU74}, Even \cite{Ev79}, Gibbons \cite{Gibbons85}, Knuth \cite{Knuth93}, Manber \cite{Manber89} and Sedgewick \cite{Sedgewick88} present versions of Algorithm~T. Dijkstra \cite{Dijkstra76} (in a monograph, not a textbook), Mehlhorn and N{\"a}her \cite{MeNa99}, and Sanders et al. \cite{SMDR19} present versions of Algorithm~C. Aho et al. \cite{AHU83}, Cormen et al. \cite{CLRS09}, Even \cite{EvEv12}, Sedgewick and Wayne \cite{SeWa14} and Dasgupta et al. \cite{DPV08} present Algorithm~B. Erickson \cite{Erickson19} presents both T and B. Finally, Sedgewick \cite{Sedgewick01C++} presents versions of all three algorithms.  Recent books tend to favor Algorithm~B, presumably because it is simpler to state and to prove correct.  But if one actually wants to compute the strong components of a graph, we favor Algorithm~T, since it is more efficient by constant factors in both space and time than the other two algorithms and easy to implement. 

For additional work on reducing the extra space needed to run depth-first search, assuming some specialized representations of the input graphs, see Hagerup~\cite{Hagerup20} and the references therein.

\section*{Acknowledgement}

We thank Don Knuth and an anonymous referee for their close and careful reading of this paper and for their insightful comments, which helped us greatly improve the presentation. 

\bibliographystyle{alpha}
\bibliography{bibliography}

\newcommand{\etalchar}[1]{$^{#1}$}
\begin{thebibliography}{GKR{\etalchar{+}}91}

\bibitem[AGU72]{AGU72}
Alfred~V. Aho, M.~R. Garey, and Jeffrey~D. Ullman.
\newblock The transitive reduction of a directed graph.
\newblock {\em {SIAM} J. Comput.}, 1(2):131--137, 1972.

\bibitem[AHU74]{AHU74}
A.V. Aho, J.E. Hopcroft, and J.D. Ullman.
\newblock {\em The design and analysis of computer algorithms}.
\newblock Addison-Wesley, 1974.

\bibitem[AHU83]{AHU83}
Alfred~V. Aho, John~E. Hopcroft, and Jeffrey~D. Ullman.
\newblock {\em Data structures and algorithms}.
\newblock Addison-Wesley, 1983.

\bibitem[ATG{\etalchar{+}}14]{ATGRZ14}
Stephen Alstrup, Mikkel Thorup, Inge~Li G{\o}rtz, Theis Rauhe, and Uri Zwick.
\newblock Union-find with constant time deletions.
\newblock {\em {ACM} Trans. Algorithms}, 11(1):6:1--6:28, 2014.

\bibitem[BDK09]{BDK09}
Piotr Berman, Bhaskar DasGupta, and Marek Karpinski.
\newblock Approximating transitive reductions for directed networks.
\newblock In {\em Proc.\ of 11th WADS}, pages 74--85, 2009.

\bibitem[CCL{\etalchar{+}}19]{CCLMT19}
Ran Chen, Cyril Cohen, Jean{-}Jacques L{\'{e}}vy, Stephan Merz, and Laurent
  Th{\'{e}}ry.
\newblock Formal proofs of {Tarjan}'s strongly connected components algorithm
  in {Why3}, {Coq} and {Isabelle}.
\newblock In {\em Proc.\ of 10th ITP}, pages 13:1--13:19, 2019.

\bibitem[CLRS09]{CLRS09}
Thomas~H. Cormen, Charles~E. Leiserson, Ronald~L. Rivest, and Clifford Stein.
\newblock {\em Introduction to algorithms}.
\newblock MIT press, 3nd edition, 2009.

\bibitem[CM96]{ChMe96}
Joseph Cheriyan and Kurt Mehlhorn.
\newblock Algorithms for dense graphs and networks on the random access
  computer.
\newblock {\em Algorithmica}, 15(6):521--549, 1996.

\bibitem[Dij75]{Dijkstra73}
Edsger~W. Dijkstra.
\newblock Finding the maximum strong components in a directed graph.
\newblock Manuscript EWD 376. Available at
  \url{https://www.cs.utexas.edu/users/EWD/
  transcriptions/EWD03xx/EWD376.html}, 1975.

\bibitem[Dij76]{Dijkstra76}
Edsger~W. Dijkstra.
\newblock {\em A Discipline of Programming}.
\newblock Prentice-Hall, 1976.

\bibitem[Dij82]{Dijkstra82}
Edsger~W. Dijkstra.
\newblock {\em Selected Writings on Computing: A Personal Perspective}.
\newblock Springer-Verlag, 1982.

\bibitem[DPV08]{DPV08}
Sanjoy Dasgupta, Christos~H Papadimitriou, and Umesh~Virkumar Vazirani.
\newblock {\em Algorithms}.
\newblock McGraw-Hill, 2008.

\bibitem[DR78]{DuRe78}
Iain~S. Duff and John~K. Reid.
\newblock An implementation of {Tarjan}'s algorithm for the block
  triangularization of a matrix.
\newblock {\em {ACM} Trans. Math. Softw.}, 4(2):137--147, 1978.

\bibitem[Edm65]{Edmonds65}
Jack Edmonds.
\newblock Paths, trees, and flowers.
\newblock {\em Canadian Journal of mathematics}, 17:449--467, 1965.

\bibitem[EE12]{EvEv12}
Shimon Even and Guy Even.
\newblock {\em Graph algorithms}.
\newblock Cambridge University Press, 2012.

\bibitem[EK77]{EvKu77}
J.~Eve and Reino Kurki{-}Suonio.
\newblock On computing the transitive closure of a relation.
\newblock {\em Acta Informatica}, 8:303--314, 1977.

\bibitem[Eri19]{Erickson19}
Jeff Erickson.
\newblock {\em Algorithms}.
\newblock Independently published, 2019.
\newblock Freely available at
  \url{https://jeffe.cs.illinois.edu/teaching/algorithms/}.

\bibitem[Eve79]{Ev79}
Shimon Even.
\newblock {\em Graph algorithms}.
\newblock Computer Science Press, 1979.

\bibitem[FM99]{FiMe99}
Ulrich Finkler and Kurt Mehlhorn:.
\newblock Checking priority queues.
\newblock In {\em Proc.\ of 10th SODA}, pages 901--902, 1999.

\bibitem[FS89]{FrSa89}
Michael~L. Fredman and Michael~E. Saks.
\newblock The cell probe complexity of dynamic data structures.
\newblock In {\em Proc.\ of 21st STOC}, pages 345--354, 1989.

\bibitem[Gab00]{Gabow00}
Harold~N. Gabow.
\newblock Path-based depth-first search for strong and biconnected components.
\newblock {\em Inf. Process. Lett.}, 74(3-4):107--114, 2000.

\bibitem[Gib85]{Gibbons85}
Alan Gibbons.
\newblock {\em Algorithmic graph theory}.
\newblock Cambridge university press, 1985.

\bibitem[GKR{\etalchar{+}}91]{GKRST91}
Phillip~B. Gibbons, Richard~M. Karp, Vijaya Ramachandran, Danny Soroker, and
  Robert~Endre Tarjan.
\newblock Transitive compaction in parallel via branchings.
\newblock {\em J. Algorithms}, 12(1):110--125, 1991.

\bibitem[GT85]{GaTa85}
Harold~N. Gabow and Robert~Endre Tarjan.
\newblock A linear-time algorithm for a special case of disjoint set union.
\newblock {\em J. Comput. Syst. Sci.}, 30(2):209--221, 1985.

\bibitem[Hag20]{Hagerup20}
Torben Hagerup.
\newblock Space-efficient {DFS} and applications to connectivity problems:
  Simpler, leaner, faster.
\newblock {\em Algorithmica}, 82(4):1033--1056, 2020.

\bibitem[HT73a]{HoTa73}
John~E. Hopcroft and Robert~Endre Tarjan.
\newblock Dividing a graph into triconnected components.
\newblock {\em {SIAM} J. Comput.}, 2(3):135--158, 1973.

\bibitem[HT73b]{HoTa73H}
John~E. Hopcroft and Robert~Endre Tarjan.
\newblock Efficient algorithms for graph manipulation {[H]} (algorithm 447).
\newblock {\em Commun. {ACM}}, 16(6):372--378, 1973.

\bibitem[HT74]{HoTa74}
John~E. Hopcroft and Robert~Endre Tarjan.
\newblock Efficient planarity testing.
\newblock {\em Journal of the ACM}, 21(4):549--568, 1974.

\bibitem[Jia93]{Jiang93}
Bin Jiang.
\newblock {I/O}-and {CPU}-optimal recognition of strongly connected components.
\newblock {\em Inf. Process. Lett.}, 45(3):111--115, 1993.

\bibitem[Kel20]{Kelly20a}
Terence Kelly.
\newblock Efficient graph search.
\newblock {\em {ACM} Queue}, 18(4):25--36, 2020.

\bibitem[KM74]{KaMu74}
Tiko Kameda and J.~Ian Munro.
\newblock A {$O(|V|{\cdot}|E|)$} algorithm for maximum matching of graphs.
\newblock {\em Computing}, 12(1):91--98, 1974.

\bibitem[Knu93]{Knuth93}
Donald~Ervin Knuth.
\newblock {\em The {Stanford GraphBase}: a platform for combinatorial
  computing}.
\newblock ACM Press, 1993.

\bibitem[Knu14]{Knuth14}
Donald~E. Knuth.
\newblock Twenty questions for {Donald Knuth}, May 20, 2014.

\bibitem[Knu21]{Knuth21}
Donald~E. Knuth.
\newblock {\em The art of computer programming}, volume 4, Pre-fascicle 12A,
  Combinatorial Algorithms.
\newblock Addison-Wesley, 2021.

\bibitem[K{\"o}n90]{Konig90}
D{\'e}nes K{\"o}nig.
\newblock {\em Theory of finite and infinite graphs}.
\newblock Birkhäuser, 1990.
\newblock English translation of ``Theorie der endlichen und unendlichen
  Graphen'', Akademische Verlagsgesellschaft, Leipzig 1936.

\bibitem[KRY95]{KRY95}
Samir Khuller, Balaji Raghavachari, and Neal~E. Young.
\newblock Approximating the minimum equivalent digraph.
\newblock {\em {SIAM} J. Comput.}, 24(4):859--872, 1995.

\bibitem[KT05]{KlTa05}
Jon Kleinberg and Éva Tardos.
\newblock {\em Algorithm design}.
\newblock Pearson, 2005.

\bibitem[Les13]{Leslie}
Tim Leslie.
\newblock Strongly connected components - an optimal algorithm as implemented
  in {SciPy}.
\newblock \url{https://www.timl.id.au/scc}, 2013.

\bibitem[Luc82]{Lucas82}
{\'E}douard Lucas.
\newblock {\em R{\'e}cr{\'e}ations math{\'e}matiques}.
\newblock Gauthier-Villars et fils, 1882.

\bibitem[Man89]{Manber89}
Udi Manber.
\newblock {\em Introduction to algorithms: a creative approach}.
\newblock Addison-Wesley, 1989.

\bibitem[MN99]{MeNa99}
Kurt Mehlhorn and Stefan N{\"a}her.
\newblock {\em {LEDA}: A platform for combinatorial and geometric computing}.
\newblock Cambridge university press, 1999.

\bibitem[Mun71]{Munro71b}
J.~Ian Munro.
\newblock Efficient determination of the transitive closure of a directed
  graph.
\newblock {\em Inf. Process. Lett.}, 1(2):56--58, 1971.

\bibitem[NS94]{NuSo94}
Esko Nuutila and Eljas Soisalon{-}Soininen.
\newblock On finding the strongly connected components in a directed graph.
\newblock {\em Inf. Process. Lett.}, 49(1):9--14, 1994.

\bibitem[Ovi98]{Ovid}
Ovid.
\newblock {\em Metamorphoses}.
\newblock Oxford University Press, 1998.
\newblock Translated by Alan David Melville.

\bibitem[Pea16]{Pearce16}
David~J. Pearce.
\newblock A space-efficient algorithm for finding strongly connected
  components.
\newblock {\em Inf. Process. Lett.}, 116(1):47--52, 2016.

\bibitem[Pur70]{Purdom70}
Paul Purdom.
\newblock A transitive closure algorithm.
\newblock {\em BIT Numerical Mathematics}, 10(1):76--94, 1970.

\bibitem[Sed88]{Sedgewick88}
Robert Sedgewick.
\newblock {\em Algorithms}.
\newblock Addison-Wesley, 2nd edition, 1988.

\bibitem[Sed01]{Sedgewick01C++}
Robert Sedgewick.
\newblock {\em Algorithms in {C++}}.
\newblock Addison-Wesley, 2001.

\bibitem[Sha81]{Sharir81}
Micha Sharir.
\newblock A strong-connectivity algorithm and its applications in data flow
  analysis.
\newblock {\em Computers \& Mathematics with Applications}, 7(1):67--72, 1981.

\bibitem[Sim89]{Simon89}
Klaus Simon.
\newblock Finding a minimal transitive reduction in a strongly connected
  digraph within linear time.
\newblock In {\em Proc.\ of 15th WG}, pages 245--259, 1989.

\bibitem[SMDD19]{SMDR19}
Peter Sanders, Kurt Mehlhorn, Martin Dietzfelbinger, and Roman Dementiev.
\newblock {\em Sequential and Parallel Algorithms and Data Structures}.
\newblock Springer, 2019.

\bibitem[SW64]{SaWe64}
R.W.H. Sargent and A.W. Westerberg.
\newblock {SPEED-UP} in chemical engineering design.
\newblock {\em Trans. Inst. Chem. Eng}, 42:190--197, 1964.

\bibitem[SW14]{SeWa14}
Robert Sedgewick and Kevin Wayne.
\newblock {\em Algorithms}.
\newblock Addison-Wesley, 4th edition, 2014.

\bibitem[Tar72]{Tarjan72}
Robert~E. Tarjan.
\newblock Depth-first search and linear graph algorithms.
\newblock {\em SIAM Journal on Computing}, 1(2):146--160, 1972.

\bibitem[Tar75]{Tarjan75}
Robert~E. Tarjan.
\newblock Efficiency of a good but not linear set union algorithm.
\newblock {\em Journal of the ACM}, 22(2):215--225, 1975.

\bibitem[Tar83a]{Tarjan83}
Robert~E. Tarjan.
\newblock {\em Data structures and network algorithms}.
\newblock SIAM, 1983.

\bibitem[Tar83b]{Tarjan83b}
Robert~Endre Tarjan.
\newblock Space-efficient implementations of graph search methods.
\newblock {\em {ACM} Trans. Math. Softw.}, 9(3):326--339, 1983.

\bibitem[Tra63]{Trakhtenbrot63}
Boris~Abramovich Trakhtenbrot.
\newblock {\em Algorithms And Automatic Computing Machines (Topics In
  Mathematics)}.
\newblock D.C. Heath and Company, 1963.
\newblock Translated from Russian.

\bibitem[Vet01]{Vetta01}
Adrian Vetta.
\newblock Approximating the minimum strongly connected subgraph via a matching
  lower bound.
\newblock In {\em Proc.\ of 12th SODA}, pages 417--426. {ACM/SIAM}, 2001.

\bibitem[YS14]{YuSh14}
Chengcheng Yu and Zhonge Sheng.
\newblock A matching algorithm based on the depth first search for the general
  graph.
\newblock In {\em 2nd International Conference on Information, Electronics and
  Computer}, pages 59--63. Atlantis Press, 2014.

\bibitem[ZNI03]{ZNI03}
Liang Zhao, Hiroshi Nagamochi, and Toshihide Ibaraki.
\newblock A linear time 5/3-approximation for the minimum strongly-connected
  spanning subgraph problem.
\newblock {\em Inf. Process. Lett.}, 86(2):63--70, 2003.

\end{thebibliography}

\appendix
\newpage

\section{Quick Search versus Depth-First Search}\label{A-quick}

A search algorithm often confused with depth-first search is an algorithm that Knuth \cite[Section 7.4.1.2, Algorithm Q]{Knuth21} calls \emph{quick search}. Quick search is an analog of breadth-first search that uses a stack~$S$ of vertices instead of a queue.  To explore a graph it does the following:
Mark all vertices unvisited, initialize~$S$ to be empty, and repeat the appropriate one of the following steps until all vertices are visited and~$S$ is empty:

\begin{itemize}
\item[(i)] $S$ is empty and some vertex is unvisited: Let $s$ be any unvisited vertex.  Push $s$ onto~$S$. This initiates a search starting from $s$.

\item[(ii)] $S$ is non-empty:  Pop the top vertex, say $v$, from~$S$.  Visit $v$.  Traverse each arc out of~$v$. To traverse an arc from~$v$ to~$w$, if $w$ is unvisited and not on~$S$, push $w$ onto~$S$.
\end{itemize}

If one replaces the stack~$S$ in quick search by a queue, one obtains the standard implementation of breadth-first search.

Quick search does \emph{not} in general visit the vertices in the same order (either preorder or postorder) as any depth-first search, as one can easily show with appropriate counterexamples.  One can use quick search for simple tasks such as building a spanning forest, but it does not seem to have nice properties that lead to fast algorithms for many more-complicated tasks such as finding strong components.  It is faster than either of the implementations in Section~\ref{S-non-recursive-dfs}, however: It needs at most $3m+5n$ memory accesses~\cite{Knuth21}, saving $4n$ over Implementation A of depth-first search.  The same $3m+5n$ bound holds for the corresponding implementation of breadth-first search.  One can use either quick search or breadth-first search to do the backward search in Algorithm~B, the bidirectional algorithm to find strong components.  

It is possible to modify quick search to simulate the effect of a depth-first search, but doing so is not straightforward and produces implementations that are more complicated and less efficient than either Implementation~V or Implementation~A.  One must either store multiple copies of vertices on the stack, the copies being the tips of traversible arcs; or, store only one copy of each vertex, but move it to the top of the stack when a new entering traversible arc is discovered.  The former approach is what Dijkstra used in his implementation of Algorithm~C.  It also appears in a textbook~\cite{KlTa05}.  In both cases the simulation is incomplete, because it does not specify when vertices are postvisited, which is required in all three strong components algorithms and in other applications.  Dijkstra overcame this problem by using an additional stack and the properties of Algorithm~C; the textbook presentation ignores the issue.  The latter approach, of moving vertices to the top of the stack, was proposed by Jiang~\cite{Jiang93} who used it to implement Algorithm~T.  He claimed that the method is I/O-efficient, because vertex scans proceed without interruption. This is true of vertex scans if lists of outgoing arcs are stored in arrays and not in linked lists, but the overall algorithm is I/O-efficient only if main memory is of size $\Omega(n)$.

\end{document}